\documentclass[onecolumn, conference,a4paper]{IEEEtran}
\pdfoutput=1

\usepackage{pgfplots}
\pgfplotsset{compat=newest}
\usepackage{tikz}
\usetikzlibrary{arrows,matrix,positioning,patterns}
\usepackage[utf8]{inputenc}
\usepackage[innermargin=0.545in,outermargin=0.65in,top=0.545in,bottom=.7in]{geometry}
\usepackage[english]{babel}
\usepackage[T1]{fontenc}
\usepackage{epsfig}
\usepackage{amsmath, amssymb, amsbsy}
\usepackage{mathdots}
\usepackage{xspace}
\usepackage[noend]{algpseudocode}
\usepackage[linesnumbered,ruled,vlined,titlenumbered]{algorithm2e}
\usepackage{algorithmicx}
\usepackage{color}
\usepackage{cite}
\usepackage{booktabs}
\usepackage{verbatim}
\usepackage[colorlinks=true,citecolor=blue,linkcolor=blue,urlcolor=blue]{hyperref}
\usepackage{url}
\usepackage{lipsum}
\usepackage{enumitem}
\usepackage{colortbl}
\usepackage{amsthm}

\makeatletter
\newcommand\footnoteref[1]{\protected@xdef\@thefnmark{\ref{#1}}\@footnotemark}
\makeatother

\newtheorem{theorem}{Theorem}
\newtheorem{lemma}[theorem]{Lemma}

\newtheorem{remark}[theorem]{Remark}

\newtheorem{definition}{Definition}

\usepackage[final,tracking=true,kerning=true,spacing=true,factor=1100,stretch=10,shrink=20]{microtype}

\newenvironment{mymatrix}{\begin{bmatrix}} {\end{bmatrix} }

\def\ve#1{{\mathchoice{\mbox{\boldmath$\displaystyle #1$}}%
              {\mbox{\boldmath$\textstyle #1$}}%
              {\mbox{\boldmath$\scriptstyle #1$}}%
              {\mbox{\boldmath$\scriptscriptstyle #1$}}}}

\newcommand{\MoormatExplicit}[3]{
	\begin{mymatrix}
		#1_{1} & #1_{2} & \dots& #1_{#3}\\
		#1_{1}^{[1]} & #1_{2}^{[1]} & \dots& #1_{#3}^{[1]}\\
		\vdots &\vdots&\ddots& \vdots\\
		#1_{1}^{[#2-1]} & #1_{2}^{[#2-1]} & \dots& #1_{#3}^{[#2-1]}\\
	\end{mymatrix}}

\newcommand{\Fq}{\ensuremath{\mathbb{F}_q}}
\newcommand{\Fqm}{\ensuremath{\mathbb{F}_{q^m}}}

\DeclareMathOperator{\extsmallfield}{ext}
\DeclareMathOperator{\rank}{rk}

\newcommand{\Code}{\mathcal{C}}
\newcommand{\IC}{\mathcal{IC}}

\newcommand{\rspace}[2]{\mathcal{R}_{#2} \begin{pmatrix}#1\end{pmatrix}}

\newcommand{\rref}[1]{\text{ref}(#1) }
\renewcommand{\ker}{\mathcal{K} }

\newcommand{\mycode}[1]{\ensuremath{\mathcal{#1}}}

\newcommand{\Gabcode}[2]{\ensuremath{\mathcal{G}(#1,#2)}}

\newcommand{\IntCode}[1]{\ensuremath{\mathcal{IC}[#1]}}

\newcommand{\0}{\ve{0}}
\renewcommand{\S}{\ve{S}}
\renewcommand{\H}{\ve{H}}

\renewcommand{\a}{\ve{a}}
\renewcommand{\b}{\ve{b}}
\newcommand{\e}{\ve{e}}
\renewcommand{\v}{\ve{v}}
\newcommand{\g}{\ve{g}}

\renewcommand{\c}{\ve{c}}
\renewcommand{\e}{\ve{e}}

\newcommand{\G}{\ve{G}}
\renewcommand{\P}{\ve{P}}
\newcommand{\R}{\ve{R}}

\newcommand{\M}{\ve{M}}
\newcommand{\E}{\ve{E}}
\newcommand{\C}{\ve{C}}

\newcommand{\Hsub}{\ve{H}_\mathrm{sub}}
\newcommand{\Hsubi}{\ve{H}_{\mathrm{sub},i}}
\newcommand{\HsubEins}{\ve{H}_{\mathrm{sub},1}}
\newcommand{\HsubNKT}{\ve{H}_{\mathrm{sub},n-k-t}}

\newcommand{\A}{\ve{A}}
\newcommand{\B}{\ve{B}}
\newcommand{\Bp}{\ve{B}^{\prime }}
\newcommand{\Bppt}{\ve{B}^{\prime \prime \top}}
\newcommand{\Bpp}{\ve{B}^{\prime \prime}}
\newcommand{\Hp}{\ve{H}^{\prime}}
\newcommand{\hp}{\ve{h}^{\prime}}
\newcommand{\h}{\ve{h}}
\newcommand{\Htil}{\tilde{\ve{\H}}}
\newcommand{\I}{\ve{I}}
\newcommand{\J}{\ve{J}}
\newcommand{\D}{\ve{D}}
\newcommand{\Dp}{\ve{D}^{\prime}}

\newcommand{\dR}{\mathrm{d}_\mathrm{R}}

\newcommand{\Matmap}{\varphi}

\newcommand{\NMtnm}{\mathrm{NM}_{t,n,m}}

\newcommand{\suppR}{\mathrm{supp}_\mathrm{R}}
\newcommand{\suppH}{\mathrm{supp}_\mathrm{H}}
\newcommand{\ranksupp}[1]{\suppR(#1)}

\newcommand{\Fqmell}{\mathbb{F}_{q^{m \ell}}}

\newcommand{\tmax}{t_\mathrm{max}}
\newcommand{\Z}{\ve{Z}}
\newcommand{\F}{\mathbb{F}}

\makeatletter
\newcommand{\removelatexerror}{\let\@latex@error\@gobble}
\makeatother

\newcommand{\printalgoIEEE}[1]
{{\centering
\scalebox{0.97}{
\removelatexerror
\begin{tabular}{p{\columnwidth}}
\begin{algorithm}[H]
 \begin{small}
 #1
 \end{small}
\end{algorithm}
\end{tabular}
}
}
}

\IEEEoverridecommandlockouts

\begin{document}

\title{Decoding High-Order \\ Interleaved Rank-Metric Codes}
\author{Sven Puchinger,~\IEEEmembership{Member,~IEEE,} \IEEEauthorblockN{Julian Renner, Antonia Wachter-Zeh,~\IEEEmembership{Member,~IEEE}}
\IEEEauthorblockA{
\thanks{
  S.~Puchinger, J.~Renner, and A.~Wachter-Zeh are with the Institute for Communications Engineering, Technical University of Munich (TUM), Germany, e-mail: \{sven.puchinger, julian.renner, antonia.wachter-zeh\}@tum.de.}
\thanks{This project has received funding from the German Israeli Project Cooperation (DIP) grant no.~KR3517/9-1 and from the European Research Council (ERC) under the European Union’s Horizon 2020 research and innovation programme (grant agreement No~801434).
}
}}

\maketitle

\begin{abstract}
This paper presents an algorithm for decoding homogeneous interleaved codes of high interleaving order in the rank metric.
The new decoder is an adaption of the Hamming-metric decoder by Metzner and Kapturowski (1990) and \emph{guarantees} to correct all rank errors of weight up to $d-2$ whose rank over the large base field of the code equals the number of errors, where $d$ is the minimum rank distance of the underlying code.
In contrast to previously-known decoding algorithms, the new decoder works for \emph{any} rank-metric code, not only Gabidulin codes. It is purely based on linear-algebraic computations, and has an explicit and easy-to-handle success condition.
Furthermore, a lower bound on the decoding success probability for random errors of a given weight is derived. The relation of the new algorithm to existing interleaved decoders in the special case of Gabidulin codes is given.
\end{abstract}

\begin{IEEEkeywords}
Decoding, Gabidulin Codes, Interleaved Codes, Metzner--Kapturowski Algorithm, Rank-Metric Codes
\end{IEEEkeywords}

\section{Introduction}

Interleaved codes are direct sums of codes of the same length, where the summands are termed \emph{constituent codes} and their number is called \emph{interleaving order}.
By assuming that errors occur in certain patterns, it is possible to correct more errors than half the minimum distance.

In the Hamming metric, interleaved codes have been considered for replicated file disagreement location \cite{metzner1990general}, correcting burst errors in data-storage applications \cite{krachkovsky1997decoding}, suitable outer codes in concatenated codes \cite{metzner1990general,krachkovsky1998decoding,haslach1999decoding,justesen2004decoding,schmidt2005interleaved,schmidt2009collaborative}, an ALOHA-like random-access scheme \cite{haslach1999decoding}, decoding non-interleaved codes beyond half-the-minimum distance by power decoding \cite{schmidt2010syndrome,kampf2014bounds,rosenkilde2018power,puchinger2019improved}, and recently for code-based cryptography \cite{elleuch2018interleaved,holzbaur2019decoding}.
In all these works, the errors are assumed to be matrices with only a few non-zero columns which are added to an interleaved codeword matrix where each row is a codeword of the constituent code.
This means that the errors affect the same positions in the constituent codes (burst errors) and the number of errors is given by the number of non-zero columns of the error matrix.

There exist several decoding algorithms for interleaved Reed--Solomon codes that, for interleaving order at least two, decode beyond half the minimum distance and also beyond the Johnson radius \cite{krachkovsky1997decoding,bleichenbacher2003decoding,coppersmith2003reconstructing,parvaresh2004multivariate,brown2004probabilistic,parvaresh2007algebraic,schmidt2007enhancing,schmidt2009collaborative,cohn2013approximate,nielsen2013generalised,wachterzeh2014decoding,puchinger2017irs,yu2018simultaneous}. Beyond the unique decoding radius, decoding sometimes fails, but with a small probability
which can be bounded from above and roughly estimated by $1/q$ where $q$ is the field size of the constituent code.

Already in 1990, Metzner and Kapturowski \cite{metzner1990general} introduced a decoding algorithm for interleaved codes in the Hamming metric, where the constituent codes are the same (\emph{homogeneous interleaved codes}) and have minimum distance $d$.
The decoder can correct up to $d-2$ errors, given that the interleaving order is high enough (i.e., at least the number of errors) and that the rank of the error matrix equals the number of errors.
We want to stress that this decoding algorithm works for interleaved codes with arbitrary constituent codes, is purely based on linear-algebraic operations (i.e., row operations on matrices), and has complexity quadratic in the code length and linear in the interleaving order.
This is remarkable since the code can correct most error patterns up to almost the minimum distance of the code without assuming any side information about the error (e.g., as for erasures, where the error positions are known).
The result by Metzner and Kapturowski was later independently rediscovered in \cite{haslach1999decoding} and generalized to dependent errors by the same authors in \cite{haslach2000efficient}.

Rank-metric codes are sets of vectors over an extension field, whose elements can be interpreted as matrices over a subfield and whose distance is given by the rank of their difference.
The codes were independently introduced in \cite{Delsarte_1978,Gabidulin_TheoryOfCodes_1985,Roth_RankCodes_1991}, together with their most famous code class, Gabidulin codes, which can be seen as the rank-metric analogs of Reed--Solomon codes.

Interleaved codes in the rank metric were introduced in \cite{loidreau2006decoding} and \cite{silva2008rank}, and have found applications in code-based cryptography \cite{faure2006new,Gaborit-KeyRecoveryFaureLoidreau,wachter2018repairing,renner2018rank}, network coding \cite{silva2008rank,sidorenko2010decoding}, and construction and decoding of space-time codes \cite{gabidulin2000space,lusina2002spacetime,liu2002rank,lusina2003,robert2015new,puchinger2016space}.

Similar to the Hamming metric, in the rank metric, the errors occur as additive matrices, but their structure is different: the row spaces 
of the constituent errors
are contained in a relatively small joint row space whose dimension is the number of errors.
This (joint) row space is usually seen as the rank-metric analog to the \emph{support} of an error
\cite{Gabidulin_TheoryOfCodes_1985,Roth_RankCodes_1991,chabaud1996rsd,ourivski2002new,gaborit2016decoding,aragon:ISIT18}.

There are several algorithms for decoding interleaved Gabidulin codes \cite{loidreau2006decoding,sidorenko2010decoding,wachter2014list}, as well as efficient variants thereof \cite{sidorenko2011skew,sidorenko2014fast,puchinger2017row,puchinger2017alekhnovich}, which are able to correct most (but not all) error patterns up to a certain number of errors that is beyond half the minimum distance for interleaving orders at least two.

In this paper, we adapt Metzner and Kapturowski's algorithm to the rank metric.
As a result, we obtain an algorithm that can correct up to $d-2$ rank errors with a homogeneous interleaved code over an arbitrary constituent code of minimum rank distance $d$. The success conditions are the same as in Hamming metric: the interleaving order must be large enough and the rank of the error matrix (over the extension field) must be equal to the number of errors.
The new algorithm also works for arbitrary linear rank-metric codes, including, but not limited to, Gabidulin \cite{Delsarte_1978,Gabidulin_TheoryOfCodes_1985,Roth_RankCodes_1991}, generalized Gabidulin \cite{roth1996tensor,kshevetskiy2005new,augot2013rank}, low-rank-parity-check (LRPC) \cite{gaborit2013low}, Loidreau's Gabidulin-like \cite{loidreau2016evolution}, or twisted Gabidulin codes \cite{sheekey2016new,otal2016explicit} and their generalizations \cite{lunardon2015generalized,puchinger2017further}.
The algorithm is again purely based on linear-algebraic operations and has a complexity of
\begin{align*}
O^\sim\!\left( \max\{n^2\ell,n^3\} m \right)
\end{align*}
operations over the subfield $\Fq$, where $O^\sim$ neglects log factors, $n$ is the code length, $\ell$ is the interleaving degree, and $m$ is the extension degree of the extension field $\Fqm$ over the subfield~$\Fq$. We prove that for random errors of a given weight and growing interleaving order, the success probability gets arbitrarily close to $1$. Further, we derive sufficient conditions on the error for which the decoder is able to correct more than $d-2$ errors and present an adaption to certain heterogeneous codes. In addition, we show that by viewing a homogeneous interleaved code as a linear code over a large extension field, one obtains a (non-interleaved) linear rank-metric code and the proposed decoder corrects almost any error of rank weight up to $d-2$ in this code. Finally, we prove that in the case of Gabidulin codes, the new decoder succeeds under the same conditions as the known decoding algorithms.

The structure of this paper is as follows. In Section~\ref{sec:preliminaries}, we introduce notation, give definition,s and recall the Hamming-metric algorithm by Metzner and Kapturowski. In Section~\ref{sec:new_algorithm}, we propose the new algorithm, prove its correctness, analyze its complexity, compare it to the algorithm in Hamming metric, and give an example. In Section~\ref{sec:further_results}, we show further results including the success probability of the new decoder for random errors, sufficient conditions to successfully decode more than $d-2$ errors, an adaptation to heterogeneous interleaved codes, and relations to existing decoders. Conclusions and open problems are given in Section~\ref{sec:conclusion}.

\section{Preliminaries}\label{sec:preliminaries}
\subsection{Notation}
Let $q$ be a power of a prime and let
$\Fq$ denote the finite field of order $q$ and $\Fqm$ its extension field of order $q^m$.
Any element of $\Fq$ can be seen as an element of $\Fqm$ and $\Fqm$ is an $m$-dimensional vector space over $\Fq$.

We use $\Fq^{m \times n}$ to denote the set of all $m\times n$ matrices over $\Fq$ and $\Fqm^n =\Fqm^{1 \times n}$ for the set of all row vectors of length $n$ over $\Fqm$. Rows and columns of $m\times n$-matrices are indexed by $1,\dots, m$ and $1,\dots, n$, where $A_{i,j}$ is the element in the $i$-th row and $j$-th column of the matrix $\A$. The transposition of a matrix is indicated by superscript $\top$ and $\rref{\A}$ refers to a reduced row echelon form of $\A$. Further, we define the set of integers $[a:b] := \{i: a \leq i \leq b\}$ and the submatrix notation
\begin{equation*}
  \A_{[a:b],[c:d]} :=
  \begin{mymatrix}
    A_{a,c}& \hdots& A_{a,d}\\
    \vdots & \ddots& \vdots \\
    A_{b,c}& \hdots& A_{b,d}\\
  \end{mymatrix}.
\end{equation*}

Let $ \ve{\gamma} = \begin{mymatrix}\gamma_1,\gamma_2,\dots,\gamma_{m}\end{mymatrix}$ be an ordered basis of $\Fqm$ over $\Fq$. By utilizing the vector space isomorphism $\Fqm \cong \Fq^m$, we can relate each vector $\a \in \Fqm^n$ to a matrix $\A \in \Fq^{m \times n}$ according to
\begin{align*}
  \extsmallfield:\Fqm^{n} &\rightarrow \Fq^{m \times n}\label{eq:mapping_smallfield},\\
  ~\a = \begin{mymatrix}a_1,\hdots,a_n\end{mymatrix} &\mapsto \A = 
                                                       \begin{mymatrix}
                                                         A_{1,1}& \hdots& A_{1,n}\\
                                                         \vdots & \ddots& \vdots \\
                                                         A_{m,1}& \hdots& A_{m,n}\\
                                                         \end{mymatrix},
\end{align*}
where $a_j = \sum_{i=1}^{m} A_{i,j} \gamma_i$ for all $j \in [1,n]$.
Further, we extend the definition of $\extsmallfield$ to matrices by extending each row and then vertically concatenating the resulting matrices. A property that will be used in the paper is that if $\B \in \Fq^{t \times n}$ is a matrix from the small field $\Fq$ and $\v \in \Fqm^{n}$, then
$\extsmallfield(\v \B^{\top}) = \extsmallfield(\v) \B^{\top}$.

Let $\mathcal{V}$ be a vector space. By $\mathcal{V}^{\bot}$, we indicate the dual space of $\mathcal{V}$, i.e.,
\begin{equation*}
\mathcal{V}^{\bot} := \{ \v': \v'\v^{\top} = 0, \forall \v \in \mathcal{V} \}.
  \end{equation*}

In the following, let $\mathbb{F} \in \{\Fq,\Fqm\}$. 
We deliberately allow $\mathbb{F}$ to be the extension field $\Fqm$ or a subfield thereof.
Since then always $\mathbb{F} \subseteq \Fqm$, operations between elements of $\F$ and $\Fqm$ are well-defined.
This will be used several times throughout the paper.
The ($\F$-)span of vectors $\v_1,\dots,\v_l \in \Fqm^n$ is defined by the ($\F$-)vector space
\begin{equation*}
\langle\v_1,\hdots,\v_l\rangle_{\mathbb{F}} = \bigg\{\sum_{i=1}^{l}a_i\v_i \, : \, \ a_i \in \mathbb{F} \bigg\}.  
\end{equation*}
The ($\F$-)row space of a matrix $\A \in \Fqm^{m\times n}$ is the ($\F$-)vector space spanned by its rows,
\begin{equation*}
\rspace{\A}{\mathbb{F}} = \big \langle \begin{mymatrix} A_{1,1}, \hdots, A_{1,n} \end{mymatrix},\hdots,\begin{mymatrix} A_{m,1}, \hdots, A_{m,n} \end{mymatrix} \big \rangle_{\mathbb{F}} .
\end{equation*}
The (right) ($\F$-)kernel of a matrix $\A \in \Fqm^{m \times n}$ is the ($\F$-)vector space given by
\begin{equation*}
\ker_{\F}(\A) := \{ \v \in \F^{n}: \A\v^{\top} = \0 \}.
\end{equation*}
Note that in case of $\F = \Fq$, we can write and compute the $\Fq$-kernel as $\ker_{\Fq}(\A) := \{ \v \in \Fq^{n}:   \extsmallfield(\A)\v^{\top} = \0 \}$.
We define the $\Fqm$-rank of a matrix $\A \in \Fqm^{m\times n}$ to be
\begin{equation*}
\rank_{\Fqm}(\A) := \dim_{\Fqm} \left( \rspace{\A}{\Fqm} \right),
\end{equation*}
and its $\Fq$-rank as
\begin{equation*}
\rank_{\Fq}(\A) := \dim_{\Fq} \left( \rspace{ \extsmallfield(\A)}{\Fq} \right).
\end{equation*}
Note that the latter rank equals the $\Fq$-dimension of the $\Fq$-column span of the matrix $\A$ (and, obviously, of its extension $\extsmallfield(\A)$).
For the same matrix $\A$, the $\Fq$- and $\Fqm$-rank can be different.
In general, we have $\rank_{\Fqm}(\A) \leq \rank_{\Fq} (\A)$, where equality holds if and only if the reduced row echelon form of $\A$ has only entries in $\Fq$.

Further, throughout this paper, we use $[i]:=q^i$ for any integer $i \geq 0$.

\subsection{Rank-Metric Codes}

The rank norm $\rank_q(\a)$ of a vector $\a$ is the rank of the matrix representation $\A \in \Fq^{m \times n}$ over $\mathbb{F}_{q}$, i.e.,
\begin{equation*}
\rank_q(\a) := \rank_q(\A).
\end{equation*}
The rank distance between $\a$ and $\b$ (with $\A := \extsmallfield(\a)$ and $\B := \extsmallfield(\b)$) is defined by
\begin{equation*}
\dR(\a,\b):= \rank_q(\a-\b) = \rank_q(\A-\B).
\end{equation*}
A linear $[n,k,d]$ code $\mycode{C}$ over $\Fqm$ is a $k$-dimensional subspace of $\Fqm^n$ and minimum rank distance $d$, where
\begin{equation*}
d := \min_{\substack{\a,\b \in \mycode{C} \\ \a \neq \b }} \lbrace \rank_q(\a -\b) \rbrace  =  \min_{\a \in \mycode{C} \setminus \{0\} }\lbrace \rank_q(\a) \rbrace. 
\end{equation*}

Gabidulin codes are the first-known and most-studied class of rank-metric codes. They are defined as follows.

\begin{definition}[Gabidulin code, {\cite{Delsarte_1978,Gabidulin_TheoryOfCodes_1985,Roth_RankCodes_1991}}]
	A Gabidulin code $\Gabcode{n}{k}$ over $\Fqm$ of length $n \leq m$ 
	and dimension $k$ is defined by its $k \times n$ generator matrix
	\begin{equation*}
	\G = \MoormatExplicit{g}{k}{n},
	\end{equation*}
	where $\g=[g_1,g_2, \dots, g_{n}] \in \Fqm^n$ and $\rank_q(\g) = n$. 
\end{definition}
Gabidulin codes are MRD codes, i.e., $d=n-k+1$, and can decode uniquely and efficiently any error $\e \in \Fqm^n$ of rank weight $\rank_{\Fq}(\e) \leq \lfloor\frac{d-1}{2}\rfloor$.

Besides Gabidulin codes and variants therof based on different automorphisms \cite{roth1996tensor,kshevetskiy2005new,augot2013rank}, there are several other ($\Fqm$-)linear rank metric code constructions, for instance: low-rank-parity-check (LRPC) \cite{gaborit2013low}, which have applications in code-based cryptography, Loidreau's code class that modifies Gabidulin codes for cryptographic purposes \cite{loidreau2016evolution}, and twisted Gabidulin codes \cite{sheekey2016new,otal2016explicit}, which were the first general family of non-Gabidulin MRD codes. There are also generalizations of twisted Gabidulin codes \cite{lunardon2015generalized,puchinger2017further} and other example codes for some explicit parameters \cite{horlemann2015new}.

\subsection{Interleaved Codes}

In this paper, we propose a new decoding algorithm for homogeneous interleaved codes, which are defined as follows.

\begin{definition}\label{def:interleaved_codes}
Let $\Code[n,k,d]$ be a linear (rank- or Hamming-metric) code over $\Fqm$ and $\ell \in \mathbb{Z}_{>0}$ be a positive integer.
The corresponding ($\ell$-)interleaved code is defined by
\begin{equation*}
\IC[\ell; n,k,d] := \left\{ \C = \begin{bmatrix}
\c_1 \\ \c_2 \\ \vdots \\ \c_\ell
\end{bmatrix} \, : \, \c_i \in \Code \right\} \subseteq \Fqm^{\ell \times n} \ .
\end{equation*}
We call $\Code$ the constituent code and $\ell$ the interleaving order.
\end{definition}

Note that any codeword $\C \in \Fqm^{\ell \times n}$ of an interleaved code can be written as $\C = \M \G$, where $\G \in \Fqm^{k \times n}$ is a generator matrix of the constituent code $\C$ and $\M \in \Fqm^{\ell \times k}$ is a message.
This also directly implies that $\H \C^\top = \0 \in \Fqm^{n-k \times \ell}$ for any codeword $\C \in \IC$, where $\H$ is a parity-check matrix of the constituent code.

\subsection{Error Model and Support}

As an error model, we consider additive error matrices $\E \in \Fqm^{\ell \times n}$ of specific structure, depending on the chosen metric.
The goal of decoding is to recover a codeword $\C \in \IC$ from a received word
\begin{equation*}
\R = \C + \E \in \Fqm^{\ell \times n}.
\end{equation*}
We outline the error models for both Hamming and rank metric since we will often discuss analogies of the Hamming and rank case throughout the paper.
Furthermore, we recall the important notion of \emph{support} of an error.

\subsubsection{Hamming Metric}

In the Hamming metric, an error (of a non-interleaved code) of weight $t$ is a vector having exactly $t$ non-zero entries. It is natural to define the \emph{support} of the error as the set of indices of these non-zero positions, and many algebraic decoding algorithms aim at recovering the support of an error since it is easy to retrieve the error values afterwards.

For interleaved codes in the Hamming metric, errors of weight $t$ are considered to be matrices $\E \in \Fqm^{\ell \times n}$ that have exactly $t$ non-zero columns.
This means that errors occur at the same positions in the constituent codewords.
A natural generalization of the support of the error $\E$ is thus the set of indices of non-zero columns, i.e.,
\begin{equation*}
\suppH(\E) := \left\{ j \, : \, \text{$j$-th column of $\E$ is non-zero} \right\}.
\end{equation*}
The \emph{number of errors}, or \emph{Hamming weight} of the error $\E$, is then defined as the cardinality of the support.
Since $\E$ has only $t$ non-zero columns, we can decompose it into two matrices
\begin{equation}
\E = \A \B, \label{eq:Hamming_decomposition}
\end{equation}
where $\A \in \Fqm^{\ell \times t}$ consists of the non-zero columns of $\E$ and the rows of $\B \in \Fqm^{t \times n}$ are the corresponding $t$ identity vectors of the error positions.

\subsubsection{Rank Metric}

In the rank metric, an error of weight $t$, in the non-interleaved case, is a vector $\e \in \Fqm^n$, whose $\Fq$-rank (i.e., the $\Fq$-rank of its matrix representation $\extsmallfield(\e)$) is $t$.
It has been noted in the literature several times that the row (or column) space of the matrix representation $\extsmallfield(\e)$ of the error shares many important properties with the support notion in the Hamming metric, see, e.g., \cite{Gabidulin_TheoryOfCodes_1985,Roth_RankCodes_1991,chabaud1996rsd,ourivski2002new,gaborit2016decoding,aragon:ISIT18}.
We therefore define the \emph{(rank) support} of an error to be the row space of its matrix representation.
Then, the rank weight equals the dimension of its support.

In the interleaved case, an error of weight $t$ is a matrix $\E \in \Fqm^{\ell \times n}$ with $\Fq$-rank $t$, cf.~\cite{loidreau2006decoding,wachter2014list}.\footnote{The paper \cite{sidorenko2010decoding} considers a different error model, but the algorithm in \cite{sidorenko2010decoding} can be reformulated to work with the model considered here, cf.~\cite[Section~4.1]{wachter2013decoding}.}
Note that the matrix entries are in general over the large field $\Fqm$, but the rank is taken over $\Fq$.
Analog to the case of a single vector, we define the rank support of a matrix $\E \in \Fqm^{\ell \times n}$ to be the row space of the extended matrix $\extsmallfield(\E) \in \Fq^{\ell m \times n}$, i.e.,
\begin{equation*}
\suppR(\E) := \rspace{\extsmallfield(\E)}{\Fq}.
\end{equation*}
Thus, the number of errors, or rank weight of the error, equals the $\Fq$-dimension of the error's support.
Similar to \eqref{eq:Hamming_decomposition}, we can decompose the error matrix as follows.
\begin{lemma}[{see, e.g., \cite[Theorem~1]{matsaglia1974}}]\label{lem:E=AB}
Let $\E \in \Fqm^{\ell \times n}$ be an error matrix with $\rank_{\Fq}(\E) = t$. Then, it can be decomposed into
\begin{equation*}
\E = \A \B,
\end{equation*}
where $\A\in\Fqm^{\ell \times t}$ and $\B \in \Fq^{t \times n}$ both have full $\Fq$-rank $t$, cf. right part of Figure~\ref{fig:support_illustration}.
The matrix $\A$ and $\B$ are unique up to elementary $\Fq$-column and $\Fq$-row operations, respectively, and the rows of $\B$ are a basis of the error support $\suppR(\E)$.
\end{lemma}

For the two metrics, we will illustrate analogies of the notions of support and the decompositions, \eqref{eq:Hamming_decomposition} and Lemma~\ref{lem:E=AB}, in Figure~\ref{fig:support_illustration} (see Section~\ref{ssec:analogy}, a few pages ahead).

\subsection{Metzner--Kapturowski Algorithm for Decoding High-Order Interleaved Codes in the Hamming Metric}

In~\cite{metzner1990general}, Metzner and Kapturowski proposed a Hamming-metric decoding algorithm for interleaved codes with high interleaving order, i.e., $\ell\geq t_{\text{H}}$, where $t_\text{H}$ is the number of errors.
The algorithm is generic as it works with any code of minimum Hamming distance $d_\text{H}$.
It was shown that the proposed algorithm always retrieves the transmitted codeword if $t_{\text{H}}\leq d_{\text{H}}-2$ and if the non-zero columns of the error matrix are linearly independent, i.e., $\rank_{\Fqm}(\E) = t_\text{H}$. The algorithm is given in Algorithm~\ref{alg:MK}.
Furthermore, an illustration of the algorithm can be found in the left part of Figure~\ref{fig:support_illustration} (see Section~\ref{ssec:analogy}, some pages ahead),
which compares the classical Hamming-metric Metzner--Kapturowski algorithm with the new algorithm for rank-metric codes.

 \printalgoIEEE{
 \DontPrintSemicolon
 \KwIn{Parity-check matrix $\H$, received word $\R$}
 \KwOut{Transmitted codeword $\C$}

$\S \gets \H \R^{\top} \in \Fqm^{(n-k)\times\ell}$.

Determine $\P\in \Fqm^{(n-k)\times(n-k)}$ s.t. $\P \S = \rref{\S}$.

$\Hsub \gets (\P \H)_{[t_{\text{H}}+1:n-k],[1:n]} \in \Fqm^{(n-k-t_{\text{H}}) \times n}$.

Determine $\B\in\Fq^{t_{\text{H}} \times n}$ s.t. the columns of $\B$, which correspond to the zero-columns of $\Hsub$, form an identity matrix and the remaining columns of $\B$ are zero.

Determine $\A\in\Fqm^{\ell \times t_{\text{H}}} $ s.t. $(\H\B^{\top})\A^{\top} = \S$.

$\C \gets \R - \A\B \in \Fqm^{\ell \times n}$.

\Return{$\C$} 
 \caption{Metzner--Kapturowski Algorithm \cite{metzner1990general}}
 \label{alg:MK}
}

We observe that the algorithm first determines the error positions, i.e., $\mathrm{supp}_{\text{H}}(\E)$, by bringing the syndrome matrix $\S = \H \R^{\top}$ in reduced row echelon form and applying the same transformation to $\H$. The matrix $\Hsub$, which consists of the last $n-k-t_{\text{H}}$ rows of the transformed matrix $\P\H$, has then zero columns exactly at the error positions. After the error positions are determined, erasure decoding is performed.

\section{Decoding High-Order Interleaved Codes in the Rank Metric}\label{sec:new_algorithm}

In this section, we propose a new decoding algorithm for interleaved codes in the rank metric, which is an adaption of Metzner and Kapturowski's decoder to the rank metric and works under similar conditions for up to $t \leq d-2$ errors:
\begin{enumerate}
\item \textbf{High-order condition:} The interleaving order is at least the number of errors, i.e., $\ell \geq t$.
\item \textbf{Full-rank condition:} The error matrix has full $\Fqm$-rank, i.e., $\rank_{\Fqm}(\E) = \rank_{\Fq}(\E) = t$.
\end{enumerate}
In fact, the full-rank condition implies the high-order condition since the $\Fqm$-rank of a matrix $\E \in \Fqm^{\ell \times n}$ is at most $\ell$.
We will nevertheless mention both conditions for didactic reasons.

Throughout this section, we fix a rank-metric code $\Code$ over a field $\Fqm$ with parameters $[n,k,d]$ and a parity-check matrix $\H$ of $\Code$. We want to retrieve a codeword $\C$ of the homogeneous $\ell$-interleaved code $\IntCode{\ell;n,k,d}$, given the received word
\begin{align*}
\R = \C + \E \in \Fqm^{\ell \times n},
\end{align*}
where $\E$ is an error matrix of rank weight $\rank_{\Fq}(\E) = t$.

\subsection{The Error Support}

Similar to the Metzner--Kapturowski algorithm for the Hamming metric, our new algorithm is centered around retrieving the rank support of the error matrix $\E$ from the syndrome matrix $\S = \H \R^\top$.
As soon as $\suppR(\E)$ is known, we can recover the error $\E$ using Lemma~\ref{lem:get_E_from_B} below.
The method is a form of erasure correction, i.e., the rank-metric analog of computing the error values given the error positions in the Hamming metric.
For Gabidulin codes, this fact was already used in \cite{Gabidulin_TheoryOfCodes_1985,Roth_RankCodes_1991} and can be efficiently implemented by error-erasure decoders, cf.~\cite{silva2008rank,gabidulin2008error} or their fast variants \cite{silva2009fast,wachter2013fast,puchinger2018fast}.
In the general case, it has been an important ingredient of generic rank-syndrome decoding algorithms \cite{chabaud1996rsd,ourivski2002new,gaborit2016decoding,aragon:ISIT18}, which are mostly based on guessing the error support and then computing the error.
Since computing the error from its support is an important step of the new algorithm, we present the formal statement and proof, together with the resulting complexity, below for completeness.
\begin{lemma}[{see, e.g., \cite{Gabidulin_TheoryOfCodes_1985,Roth_RankCodes_1991,chabaud1996rsd,ourivski2002new,gaborit2016decoding,aragon:ISIT18}}]\label{lem:get_E_from_B}
Let $t<d$, $\B \in \Fq^{t \times n}$ be a basis of the rank support $\ranksupp{\E}$ of an error matrix $\E \in \Fqm^{\ell \times n}$, and $\S \in \Fqm^{n-k \times \ell} = \H \E^\top$ be the corresponding syndrome matrix.
Then, the error is given by $\E= \A \B$, where $\A \in \Fqm^{\ell \times t}$ is the unique solution of the linear system of equations
\begin{equation}
\S = (\H \B^\top) \A^\top. \label{eq:system_A_from_S_H_B}
\end{equation}
Thus, $\E$ can be computed in $O(\max\{\ell n^2,n^3\})$ operations in $\Fqm$ from $\ranksupp{\E}$ and $\S$.
\end{lemma}

\begin{proof}
Since $\B$ is a basis of $\ranksupp{\E}$, there must be a matrix $\A \in \Fqm^{\ell \times t}$ such that $\E = \A \B$. Since $\S = \H \E^\top$, $\A$ must fulfill $\S = (\H \B^\top) \A^\top$.
On the other hand, there can only be one matrix $\A$ fulfilling \eqref{eq:system_A_from_S_H_B} since, by \cite[Theorem~1]{Gabidulin_TheoryOfCodes_1985}, the matrix $\H \B^\top$ has $\Fqm$-rank $t$ due to $t<d$ and $\rank_{\Fq}(\B)=t$.
The multiplications $\H \B^\top$ and $\A \B$ cost $O(n^3)$ field operations, respectively, and solving the system \eqref{eq:system_A_from_S_H_B} requires $O(\max\{\ell n^2,n^3\})$ operations in $\Fqm$, which implies the complexity statement.
\end{proof}

\subsection{How to Determine the Error Support}

Our new decoding algorithm is based on retrieving the support of the error.
The error itself can then be computed using the method implied by Lemma~\ref{lem:get_E_from_B}.
In the following, we show how to obtain the error support from the syndrome and parity-check matrix.

Similar to Metzner and Kapturowski, we compute the syndrome matrix $\S$ as the product of the parity-check matrix $\H$ and the transposed received word $\R^\top = \C^\top + \E^\top$.
Due to the properties of the parity-check matrix, we obtain
\begin{align*}
\S = \H \E^\top.
\end{align*}
Then, we transform $\S$ into row echelon form.
Since $\S$ has $\Fqm$-rank at most $t$, which is smaller than its number of rows $n-k \geq d-1$, the resulting matrix has zero rows.
We apply the same row operations used to obtain the echelon form of $\S$ to the parity-check matrix and consider the matrix $\Hsub$, which consists of the rows of the resulting matrix corresponding to the zero rows of the echelon form of $\S$.
This process is illustrated in Figure~\ref{fig:PS_Hsub_illustration}.
The following sequence of statements derives the main statement of this section, Theorem~\ref{thm:kernelB}: the error support can be efficiently computed from $\Hsub$.

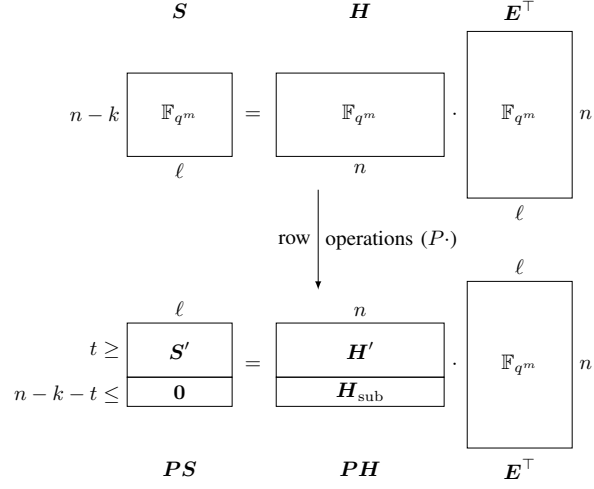
\begin{figure}
\begin{center}
  \resizebox{0.45\textwidth}{!}{
\begin{tikzpicture}[scale=0.7]
\def\nklen{2}
\def\nktlen{0.7}
\def\elllen{2.5}
\def\tlen{1.5}
\def\nlen{4}
\def\ylabeldist{2.5}
\def\xHdist{4.3}
\def\xEdist{8.1}
\def\myeps{0.5}
\def\ydist{6}

\draw (-0.5*\elllen,-0.5*\nklen) rectangle node {$\Fqm$} (0.5*\elllen,0.5*\nklen);
\node[left] at (-0.5*\elllen,0) {$n-k$};
\node[below] at (0,-0.5*\nklen) {$\ell$};
\node at (0,\ylabeldist) {$\S$};

\node at (0.5*\elllen+\myeps,0) {$=$};

\draw (\xHdist-0.5*\nlen,-0.5*\nklen) rectangle node {$\Fqm$} (\xHdist+0.5*\nlen,0.5*\nklen);
\node[below] at (\xHdist+0,-0.5*\nklen) {$n$};
\node at (\xHdist+0,\ylabeldist) {$\H$};

\node at (\xHdist+0.5*\nlen+0.5*\myeps,0) {$\cdot$};

\draw (\xEdist-0.5*\elllen,-0.5*\nlen) rectangle node {$\Fqm$} (\xEdist+0.5*\elllen,0.5*\nlen);
\node[right] at (\xEdist+0.5*\elllen,0) {$n$};
\node[below] at (\xEdist+0,-0.5*\nlen) {$\ell$};
\node at (\xEdist+0,\ylabeldist) {$\E^\top$};

\draw[->,>=latex] (\xHdist-0.25*\nlen, -0.3*\ydist) to node [left] {row}  node [right] {operations ($P \cdot$)} (\xHdist-0.25*\nlen, -0.7*\ydist);

\draw (-0.5*\elllen,-\ydist-0.5*\nklen+\nktlen) rectangle node {$\S'$} (0.5*\elllen,-\ydist+0.5*\nklen);
\draw (-0.5*\elllen,-\ydist-0.5*\nklen) rectangle node {$\0$} (0.5*\elllen,-\ydist-0.5*\nklen+\nktlen);
\node[left] at (-0.5*\elllen,-\ydist+0.5*\nktlen) {$t\geq$};
\node[left] at (-0.5*\elllen,-\ydist-\nktlen) {$n-k-t\leq$};
\node[above] at (0,-\ydist+0.5*\nklen) {$\ell$};
\node at (0,-\ydist-\ylabeldist) {$\P\S$};

\node at (0.5*\elllen+\myeps,-\ydist+0) {$=$};

\draw (\xHdist-0.5*\nlen,-\ydist-0.5*\nklen+\nktlen) rectangle node {$\H'$} (\xHdist+0.5*\nlen,-\ydist+0.5*\nklen);
\draw (\xHdist-0.5*\nlen,-\ydist-0.5*\nklen) rectangle node {$\Hsub$} (\xHdist+0.5*\nlen,-\ydist-0.5*\nklen+\nktlen);

\node[above] at (\xHdist+0,-\ydist+0.5*\nklen) {$n$};
\node at (\xHdist+0,-\ydist-\ylabeldist) {$\P\H$};
\node at (\xHdist+0.5*\nlen+0.5*\myeps,-\ydist+0) {$\cdot$};

\draw (\xEdist-0.5*\elllen,-\ydist-0.5*\nlen) rectangle node {$\Fqm$} (\xEdist+0.5*\elllen,-\ydist+0.5*\nlen);
\node[right] at (\xEdist+0.5*\elllen,-\ydist+0) {$n$};
\node[above] at (\xEdist+0,-\ydist+0.5*\nlen) {$\ell$};
\node at (\xEdist+0,-\ydist-\ylabeldist) {$\E^\top$};

\end{tikzpicture}
}
\end{center}
\caption{Illustration of Lemma~\ref{lem:Hsub_basis_general}. 
}
\label{fig:PS_Hsub_illustration}
\end{figure}

\begin{lemma}\label{lem:Hsub_basis_general}
Let $\S = \H \E^{\top} \in \Fqm^{(n-k) \times \ell}$ be the syndrome of an error $\E \in \Fqm^{\ell \times n}$ of rank weight $\rank_{\Fq}(\E)=t<n-k$ and $\P\in \Fqm^{(n-k)\times(n-k)}$ be a matrix of rank $\rank_{\Fqm}(\P)=n-k$ such that $\P\S$ is in row-echelon form.
Then, at least $n-k-t$ rows of $\P\S$ are zero.
Let $\Hsub$ be the rows of $\P\H$ corresponding to the zero rows in $\P \S$.
Then, $\Hsub$ is a basis of
\begin{align*}
\ker_{\Fqm}(\E) \cap \Code^\bot.
\end{align*}
\end{lemma}

\begin{proof}
Since $\E$ has $\Fq$-rank $t$, its $\Fqm$-rank is at most $t$. Hence, the $\Fqm$-rank of $\S$ is at most $t$, and at least $n-k-t$ of the $n-k$ rows are zero in its echelon form $\P \S$.

The rows of $\P\H$ (and thus of $\Hsub$) are in the row space of $\H$, which is equal to $\Code^\bot$.
Furthermore, the rows of $\Hsub$ are in the kernel of $\E$ since $\Hsub \E^\top = \0$.
It is left to show that the rows span the entire intersection.
Write
\begin{align*}
\P \S =
\begin{bmatrix}
\S' \\ \0
\end{bmatrix}, \quad \P\H = \begin{bmatrix}
\H' \\ \Hsub
\end{bmatrix},
\end{align*}
where $\S' = \H' \E^\top$ has full rank and has as many rows as $\H'$.
Let $\h = [\v_1,\v_2] \begin{bmatrix}
\H' \\ \Hsub
\end{bmatrix}$ be a vector in the row space of $\P\H$ and in the kernel of $\E$.
Then, we can write
\begin{align*}
\0 = \h \E^\top = [\v_1,\v_2] \begin{bmatrix}
\H' \\ \Hsub
\end{bmatrix} \E^\top = \v_1 \H' \E^\top = \v_1 \S'
\end{align*}
due to $\Hsub \E^\top = \0$. This implies that $\v_1 = \0$ since the rows of $\S'$ are linearly independent.
Thus, $\h$ is in the row space of $\Hsub$.
\end{proof}

Lemma~\ref{lem:Hsub_basis_general} shows that the matrix $\Hsub$ is connected to the kernel of the error.
The next lemma proves that this kernel is closely connected to the rank support of the error if the $\Fqm$-rank of the error is $t$ (\textbf{full-rank condition}).
Note that we only required $\rank_{\Fq}(\E)=t$ in Lemma~\ref{lem:Hsub_basis_general}, which is a weaker condition.

\begin{lemma}\label{lem:Hsub_basis_full-rank}
Let $\E \in \Fqm^{\ell \times n}$ be an error of $\Fq$-rank $t$.
If $\ell \geq t$ (\textbf{high-order condition}) and $\rank_{\Fqm}(\E) = t$ (\textbf{full-rank condition}), then
\begin{equation*}
\ker_{\Fqm}(\E) = \ker_{\Fqm}(\B),
\end{equation*}
where $\B$ is any basis of the error support $\ranksupp{\E}$.
\end{lemma}

\begin{proof}
Let $\E$ have $\Fqm$-rank $t$ (recall that the high-order condition is necessary for this) and $\E=\A \B$ be a decomposition as in Lemma~\ref{lem:E=AB}.
Then, $\A \in \Fqm^{\ell \times t}$ must have full $\Fqm$-rank $t$ and $\ker_{\Fqm}(\A)=\0$. Thus, for all $\v \in \Fqm^{n} $, $(\A \B) \v^{\top} = \0$ if and only if $\B \v^{\top} = \0$, which is equivalent to
\begin{equation*}
\ker_{\Fqm}(\E) = \ker_{\Fqm}(\A \B) = \ker_{\Fqm}(\B). \qedhere
\end{equation*}
\end{proof}

In order to prove that the error support can be computed from $\Hsub$, we require the following property of $\extsmallfield(\Hsub)$.

\begin{lemma}
  Let $\Hsub\in\Fqm^{(n-k-t)\times n}$ and $\h \in \rspace{\Hsub}{\Fqm}$. Then each row of $\extsmallfield(\h)$ is in $\rspace{\extsmallfield(\Hsub)}{\Fq}$.
  \label{lemma:extg}
  \end{lemma}

  \begin{proof}
    Since $\h \in \rspace{\Hsub}{\Fqm}$, the vector $\h$ can be written as
    \begin{equation}
      \h = \sum_{i=1}^{n-k-t} a_i \Hsubi,
      \label{eq:hlincombi}
\end{equation}
where $a_1,\hdots,a_{n-k-t} \in \Fqm$ and $\Hsubi$ denotes the $i$-th row of $\Hsub$. Using the vector and matrix representation of finite field elements, equation~(\ref{eq:hlincombi}) can be mapped to
    \begin{equation*}
      \extsmallfield(\h) = \sum_{i=1}^{m} \M_{a_i} \extsmallfield(\Hsubi),
    \end{equation*}
    where $\M_{a_i} \in \Fq^{m\times m}$ is the matrix representation of $a_i$ over $\Fq$ for a given basis $\boldsymbol{\gamma}$, cf.~\cite{wardlaw_1994}. Since $\M_{a_i}$ for $i=1,\hdots,m$ is over $\Fq$, each row of $\extsmallfield(\h)$ is in $\rspace{\extsmallfield(\Hsub)}{\Fq}$.
    \end{proof}

By combining the three lemmas above, the following theorem shows that the rank support of the error can be computed from $\Hsub$ under the \textbf{high-order} and \textbf{full-rank condition}.
In the Hamming metric, Metzner and Kapturowksi required the number of errors to be $t\leq d-2$ since they used the fact that any $t+1 \leq d-1$ columns of the parity-check matrix are linearly independent.
Here, we obtain the same condition as we use the rank-metric analog of this statement, \cite[Theorem~1]{Gabidulin_TheoryOfCodes_1985}: If the parity-check matrix is multiplied from the right by any matrix over the small field $\Fq$ of rank $t+1 \leq d-1$, then the resulting matrix has rank $t+1$.

\begin{theorem}\label{thm:kernelB}
Let $\E \in \Fqm^{\ell \times n}$ be an error of $\Fq$-rank $t\leq d-2$.
If $\ell \geq t$ (\textbf{high-order condition}) and $\rank_{\Fqm}(\E) = t$ (\textbf{full-rank condition}), then
\begin{equation*}
\ranksupp{\E} = \ker_{\Fq}(\extsmallfield(\Hsub)),
\end{equation*}
where $\Hsub$ is defined as in Lemma~\ref{lem:Hsub_basis_general}.
\end{theorem}

\begin{IEEEproof}
Lemmas~\ref{lem:Hsub_basis_general} and \ref{lem:Hsub_basis_full-rank} show that the high-order and full-rank conditions imply
\begin{equation}
\rspace{\Hsub}{\Fqm} = \ker_{\Fqm}(\B) \cap \Code^\bot. \label{eq:Hsub_kerB_dualCode}
\end{equation}
In the following, we prove that if we consider the $\Fq$-row space of the extended $\Hsub$ instead of the $\Fqm$-row space of $\Hsub$, the result is directly the $\Fq$-kernel of $\B$, i.e.,
\begin{equation*}
\rspace{\extsmallfield(\Hsub)}{\Fq} = \ker_{\Fq}(\B).
\end{equation*}
Together with $\ker_{\Fq}(\extsmallfield(\Hsub)) = \rspace{\extsmallfield(\Hsub)}{\Fq}^\bot$ and $\ranksupp{\B} = \ker_{\Fq}(\B)^\bot$, this implies the claim.

First, we prove that $\rspace{\extsmallfield(\Hsub)}{\Fq} \subseteq \ker_{\Fq} (\B)$.
It suffices to show that any row of $\extsmallfield(\Hsub)$ is in the $\Fq$-kernel of $\B$.
Such a row is again a row $\v_i$ of $\extsmallfield(\v)$.
Since obviously $\v_i \in \Fq^n$, it is left to show that $\B \v_i^\top = 0$.
This follows from \eqref{eq:Hsub_kerB_dualCode}, which implies that $\v \in \ker_{\Fqm}(\B)$ and thus
\begin{equation*}
\0 = \extsmallfield(\B \v^\top) = \B \extsmallfield(\v)^\top,
\end{equation*}
where the second equality is true since $\B$ has entries in $\Fq$.

Second, we show $\ker_{\Fq} (\B) \subseteq \rspace{\extsmallfield(\Hsub)}{\Fq}$ by proving that $r:=\dim\big(\rspace{\extsmallfield(\Hsub)}{\Fq}\big) = \dim(\ker_{\Fq} (\B))= n-t$.
Since $\rspace{\extsmallfield(\Hsub)}{\Fq} \subseteq \ker_{\Fq} (\B)$, $r>n-t$ is not possible. In the following, we show that $r<n-t$ is also not possible meaning $r$ must be equal to $n-t$.

For the proof, we will use the following mapping, which, for arbitrary $\mu,\nu \in \mathbb{N}$ and a matrix $\A \in \Fq^{\mu \times \nu}$, sends a vector $\v \in \Fqm^\mu$ to the $\Fqm$-row space of $\A$:
\begin{align*}
\Matmap_{\A} : \Fqm^{\mu} \rightarrow \Fqm^{\nu}, \quad \v \mapsto \v \A.
\end{align*}

Let $\ve{h}_1,\hdots,\ve{h}_r\in \Fq^{n}$ be a basis of $\rspace{\extsmallfield(\Hsub)}{\Fq}$, i.e.,
\begin{equation*}
\langle \ve{h}_1,\hdots,\ve{h}_r \rangle_{\Fq} = \rspace{\extsmallfield(\Hsub)}{\Fq},
\end{equation*}
and assume $r<n-t$.

According to the basis extension theorem, there is a matrix $\Bpp \in \Fq^{(n-t)\times n}$ such that
\begin{equation*}
  \Bp :=
  \begin{mymatrix}
\B^{\top} \Bppt
    \end{mymatrix} \in \Fq^{n\times n},
\end{equation*}
has $\Fq$-rank $n$. Since $\Bp$ has full rank, the mapping $\Matmap_{\Bp}$ is bijective and thus
\begin{equation*}
  \Hp=\begin{mymatrix}
    \hp_1 \\
    \vdots \\
    \hp_r
    \end{mymatrix}
    :=
    \begin{mymatrix}
    \Matmap_{\Bp}(\h_1) \\
    \vdots \\
    \Matmap_{\Bp}(\h_r)
    \end{mymatrix} \in \Fq^{n \times n}
  \end{equation*}
  has $\Fq$-rank $r$. The vectors $\h_1,\hdots,\h_r$ are in the right $\Fq$-kernel of $\B$, so
\begin{equation*}
  \Hp=
  \begin{mymatrix}
    0 & \hdots & 0 & h^{\prime}_{1,t+1} & h^{\prime}_{1,t+2} & \hdots & h^{\prime}_{1,n}  \\ 
    \vdots &\ddots & \vdots & \vdots & \vdots & \ddots & \vdots\\
    0 & \hdots & 0 & h^{\prime}_{r,t+1} & h^{\prime}_{1,t+2} & \hdots & h^{\prime}_{r,n}  \\ 
    \end{mymatrix}.
  \end{equation*}
  
  By assumption, $r<n-t$ which means that there exists a full-rank matrix $\J \in \Fq^{(n-t)\times n}$ such that
  \begin{align}
    \Htil :&= \Hp \begin{mymatrix}
      \I \\
      \J
    \end{mymatrix} \nonumber \\
    &=
    \begin{mymatrix}
    0 & \hdots & 0 & 0 & h^{\prime}_{1,t+2} & \hdots & h^{\prime}_{1,n}  \\ 
    \vdots &\ddots & \vdots & \vdots & \vdots & \ddots & \vdots\\
    0 & \hdots & 0 & 0 & h^{\prime}_{r,t+2} & \hdots & h^{\prime}_{r,n}  \\ 
    \end{mymatrix},\label{eq:Htil}
    \end{align}
    where $\I\in \Fq^{t \times n}$ is a matrix that consists of an identity matrix in the first $t$ columns and a zero matrix in the last $n-t$ columns. Further, the matrix
    \begin{equation*}
      \D :=   \begin{mymatrix}
\B^{\top} \Bppt
    \end{mymatrix} \begin{mymatrix}
      \I \\
      \J
    \end{mymatrix} \in \Fq^{n \times n}
      \end{equation*}
      has $\Fq$-rank $n$, which means that $\Matmap_{\D}$ is bijective and the matrix $\Htil$ can be written as
  \begin{equation*}
    \Htil = \begin{mymatrix}
      \Matmap_{\D}(\h_1) \\
      \vdots \\
      \Matmap_{\D} (\h_r)
    \end{mymatrix}.
      \end{equation*}
The matrix $\Dp := \D_{[1:n],[1:t+1]} \in \Fq^{n \times t+1}$ has $\Fq$-rank $t+1$ and it follows by (\ref{eq:Htil}) that
      \begin{equation}
        \Matmap_{\Dp}(\h_i) = \begin{mymatrix} 0 & \hdots & 0  &  0  \end{mymatrix} \in \Fq^{t+1}
        \label{eq:hDp}
    \end{equation}
    for $i=1,\hdots,r$.

    However, 
    \begin{equation*}
\rank_{\Fqm}(\H \Dp ) = t+1,
      \end{equation*}
      since $\H\in \Fqm^{(n-k)\times n}$ is a parity-check matrix of a $[n,k,d]$ code and  $\rank_{\Fq}(\Dp) = t+1$~\cite[Theorem 1]{Gabidulin_TheoryOfCodes_1985}. Thus, there exists a vector $\g \in \rspace{\H}{\Fqm}$ such that
      \begin{equation}
        \g \Dp =
        \begin{mymatrix}
          0& \hdots&0 &g_{t+1}^{\prime}
        \end{mymatrix},
        \label{eq:gDp}
        \end{equation}
        where $g_{t+1}^{\prime} \in \Fqm \setminus \{0\}$. Since the first $t$ positions of $\g \Dp$ are equal to zero, $\g \in \rspace{\Hsub}{\Fqm}$. From~(\ref{eq:gDp}) follows that
      \begin{equation}
        \extsmallfield(\g) \Dp =
        \begin{mymatrix}
          0& \hdots&0 &g_{1,t+1}^{\prime} \\
          0& \hdots&0 &g_{2,t+1}^{\prime} \\
          \vdots &\ddots & \vdots & \vdots \\
          0& \hdots&0 &g_{m,t+1}^{\prime} \\
        \end{mymatrix}\in \Fq^{m \times (t+1)},
        \label{eq:extgDp}
      \end{equation}
      where $\extsmallfield(g_{t+1}^{\prime}) = \begin{mymatrix} g_{1,t+1}^{\prime} \hdots g_{m,t+1}^{\prime}\end{mymatrix}^{\top} \in \Fq^{m \times 1}$.
        Let $\g_i$ denote the $i$-th row of $\extsmallfield(\g)$ for which $g_{i,t+1}^{\prime} \neq 0$. Then by equation~(\ref{eq:extgDp}),
        \begin{equation*}
          \Matmap_{\Dp}(\g_i) =
          \begin{mymatrix}
          0& \hdots&0 &g_{i,t+1}^{\prime}
        \end{mymatrix}.
\end{equation*}
Since  $\g_i \in \rspace{\extsmallfield(\Hsub)}{\Fq}$, cf. Lemma~\ref{lemma:extg}, this constitutes a contradiction according to equation~(\ref{eq:hDp}) which says that for all $\g_i \in \rspace{\extsmallfield(\Hsub)}{\Fq}$,
        \begin{equation*}
          \Matmap_{\Dp}(\g_i) =
          \begin{mymatrix}
          0& \hdots&0 & 0
        \end{mymatrix}.
\end{equation*}
Thus, $r<n-t$ is not possible and leaves $r=n-t$ as only valid possibility, which means $\rspace{\extsmallfield(\Hsub)}{\Fq} = \ker_{\Fq}(\B)$.

Recall that $\rspace{\extsmallfield(\Hsub)}{\Fq} = \ker_{\Fq}(\B)$ is equivalent to
\begin{equation*}
\ranksupp{\E} = \rspace{\B}{\Fq} = \ker_{\Fq}(\extsmallfield(\Hsub)),
\end{equation*}
which proves the claim.
\end{IEEEproof}

\subsection{The New Algorithm}

The results of the previous subsections imply an efficient decoding algorithm for high-order interleaved codes, which is summarized in Algorithm~\ref{alg:Decoding}.
We prove its correctness and state the resulting complexity in Theorem~\ref{thm:main_statement} below.

 \printalgoIEEE{
 \DontPrintSemicolon
 \KwIn{Parity-check matrix $\H$, received word $\R$}
 \KwOut{Transmitted codeword $\C$}

$\S \gets \H \R^{\top} \in \Fqm^{(n-k)\times\ell}$. \label{line:S}

Determine $\P\in \Fqm^{(n-k)\times(n-k)}$ s.t. $\P \S = \rref{\S}$. \label{line:P}

$\Hsub \gets (\P \H)_{[t+1:n-k],[1:n]} \in \Fqm^{(n-k-t) \times n}$. \label{line:Hsub}

Determine $\B\in\Fq^{t \times n}$ s.t. $ \extsmallfield(\Hsub) \B^{\top} = \0$ and $\rank_q(\B)=t$. \label{line:B}

Determine $\A\in\Fqm^{\ell \times t} $ s.t. $(\H\B^{\top})\A^{\top} = \S$. \label{line:A}

$\C \gets \R - \A\B \in \Fqm^{\ell \times n}$. \label{line:C}

\Return{$\C$} \label{line:return}
 \caption{New Decoder for High-Order Interleaved Rank-Metric Codes} 
 \label{alg:Decoding}
 }

\begin{theorem}\label{thm:main_statement}
Let $\C$ be a codeword of a homogeneous $\ell$-interleaved rank-metric code $\IntCode{\ell ; n,k,d}$ of minimum rank distance $d$. 
Furthermore, let $\E \in \Fqm^{\ell \times n}$ be an error matrix containing $\rank_{\Fq}(\E) = t \leq d-2$ errors that fulfills $t \leq \ell$ (\textbf{high-order condition}) and $\rank_{\Fqm}(\E) = t$ (\textbf{full-rank condition}). Then $\C$ can be uniquely reconstructed from the received word
using Algorithm~\ref{alg:Decoding} in
\begin{equation*}
O(\max\{n^3,n^2\ell\})
\end{equation*}
operations in $\Fqm$ plus
\begin{equation*}
O(n^3m)
\end{equation*}
operations in $\Fq$.
\end{theorem}

\begin{proof}
According to Lemma~\ref{lem:E=AB}, the error matrix can be decomposed into $\E = \A \B$. Algorithm~\ref{alg:Decoding} determines first a basis of the error's rowspace, i.e., $\B$, and then the matrix $\A$. 

To determine $\B$, the algorithm computes in Step~\ref{line:P} a transformation matrix $\P$ such that $\P\S$ is in row echelon form. In Step~\ref{line:Hsub}, $\Hsub$ is chosen as the last $n-k-t$ rows of $\P\H$.
Then, Algorithm~\ref{alg:Decoding} determines $\B$ by finding a basis of $\ker_{\Fq}(\extsmallfield(\Hsub))$ in Step~\ref{line:B}.
This is possible due to Theorem~\ref{thm:kernelB} and the condition $\rank_{\Fqm}(\E)=t$.

Having $\B$, the matrix $\A$ is obtained in Step~\ref{line:A} by solving $\S = (\H\B^{\top})\A^{\top}$ for $\A$, see Lemma~\ref{lem:get_E_from_B}. Thus, Algorithm~\ref{alg:Decoding} returns the transmitted codeword in Step~\ref{line:return}.

The complexity of the steps in the algorithm are given by
\begin{itemize}
\item Line~\ref{line:S} (Multiplication of $\H \R^\top$) requires $O(n(n-k)\ell)\subseteq O(n^2\ell)$ operations in $\Fqm$.
\item Line~\ref{line:P} (Transformation of $\begin{mymatrix} \S & \I \end{mymatrix}$ in reduced row echelon form) needs $O((n-k)^2(\ell+n-k)) \subseteq O(\max \{n^3,n^2\ell\} )$ operations in $\Fqm$.
\item Line~\ref{line:Hsub} (Multiplication of $(\P \H)_{[t+1:n-k],[1:n]}$) requires $(n(n-k-t)(n-k)) \subseteq  O(n^3)$ operations in $\Fqm$.
\item Line~\ref{line:B} (Transformation of $\begin{mymatrix} \extsmallfield(\Hsub)^{\top} & \I^{\top} \end{mymatrix}^{\top} $ in column echelon form) needs $O(n^2((n-k-t)m+n)) \subseteq O(n^3m)$ operations in $\Fq$.
\item Line~\ref{line:A} (Transformation of $\begin{mymatrix} (\H\B^{\top}) & \S \end{mymatrix} $ in row echelon form) requires $O((n-k)^2(t+\ell)) \subseteq O(\max\{n^3,n^2\ell\})$ operations in $\Fqm$.
\item Line~\ref{line:C} (Multiplication of $\A\B$) needs $O(\ell t n) \subseteq O(\ell n^2$ operations in $\Fqm$.
\item Line~\ref{line:C} (Subtraction of $\R$ and $(\A\B)$) needs $O(\ell n)$ operations in $\Fqm$.
\end{itemize}
Thus, Algorithm~\ref{alg:Decoding} requires $O(\max\{n^3,n^2\ell\})$ operations in $\Fqm$ and $O(n^3m)$ operations in $\Fq$.
\end{proof}

\begin{remark}\label{rem:complexity_in_Fq}
Using a standard (e.g., polynomial or low-complexity normal) basis of $\Fqm$ over $\Fq$ with standard algorithms (e.g., naive polynomial multiplication and division), operations in $\Fqm$ cost $O(m^2)$ operations in $\Fq$. Thus, Algorithm~\ref{alg:Decoding} has overall complexity
\begin{equation*}
O\!\left( \max\{n^3,n^2\ell\} m^2 \right)
\end{equation*}
over $\Fq$.
Asymptotically, it can be implemented faster.
The bases of $\Fqm$ over $\Fq$ presented in \cite{couveignes2009elliptic} enable to multiply, add, and $q$-power elements of $\Fqm$ in, asymptotically, $O^\sim(m)$ operations over $\Fq$, where $O^\sim$ neglects $\log$-factors in $m$.
This is done by representing field elements in suitable bases of $\Fqm$ over $\Fq$.
Hence, the overall complexity of Algorithm~\ref{alg:Decoding} can be given as
\begin{equation*}
O^\sim\!\left( \max\{n^3,n^2\ell\} m \right)
\end{equation*}
operations in $\Fq$.
Note, however, that the cost for a fixed $m$ might be larger than the standard approach due to a larger hidden constant.
\end{remark}

\subsection{Implementation of the Algorithm}

We have implemented and tested Algorithm~\ref{alg:Decoding} in the computer-algebra system SageMath v8.6~\cite{sagemath}.
The implementation and test scripts are available online under \url{https://bitbucket.org/julianrenner/ir_decoder}.

\subsection{Analogy to Metzner--Kapturowski in the Hamming Metric}\label{ssec:analogy}

The new decoder, Algorithm~\ref{alg:Decoding}, is the rank-metric analog of the Metzner--Kapturowski algorithm in the Hamming metric.
In the following, we will draw important connections between the algorithms, which once again substantiate the analogy of the support in the Hamming metric and the rank support in the rank metric.
The comparison is illustrated in Figure~\ref{fig:support_illustration}.

\begin{figure}[ht]
\begin{center}
  \resizebox{\linewidth}{!}{
\begin{tikzpicture}[scale=0.7]
\def\nlen{4.4}
\def\elllen{2.5}
\def\tlen{1.5}
\def\ylabeldist{1.7}
\def\xAdist{3.7}
\def\xBdist{7.0}
\def\myeps{0.5}
\def\ydist{6}

\def\Ewidth{0.5}
\def\xEone{1.0}
\def\xEtwo{2.5}
\def\xEthree{\xEtwo+\Ewidth}

\def\yArlen{1.5}
\def\yArlenTwo{0.25}

\def\Hsubdis{5.0}
\def\nktlen{1.5}

\def\xHsub{1.0}
\def\xHsubExt{\xHsub+6.5}

\def\matExt{2}

\def\Hsubdis{5.0}
\def\nktlen{2.2}
\def\extnktlen{4.4}
\def\Formdist{9.0}

\def\rkdist{14.0}

\def\linedist{0.5*(\rkdist+\xBdist)}

\def\xCaptionOne{3}
\def\xCaptionTwo{18}
\def\yCaption{3}

\node at (\xCaptionOne,\yCaption) {\textbf{Hamming Metric}};
\node at (\xCaptionTwo,\yCaption) {\textbf{Rank Metric}};

\draw (-0.5*\nlen,-0.5*\elllen) rectangle (0.5*\nlen,0.5*\elllen);
\node[left] at (-0.5*\nlen,0) {$\ell$};
\node[above] at (0,0.5*\elllen) {$n$};
\node at (0,-\ylabeldist) {$\E$};

\draw[pattern=north west lines] (-0.5*\nlen+\xEone,-0.5*\elllen) rectangle (-0.5*\nlen+\xEone+1*\Ewidth,0.5*\elllen);
\draw[pattern=north west lines] (-0.5*\nlen+\xEtwo,-0.5*\elllen) rectangle (-0.5*\nlen+\xEtwo+1*\Ewidth,0.5*\elllen);
\draw[pattern=north west lines] (-0.5*\nlen+\xEthree,-0.5*\elllen) rectangle (-0.5*\nlen+\xEthree+1*\Ewidth,0.5*\elllen);
\node at (0.5*\nlen+\myeps,0) {$=$};

\draw[<-,>=latex] (-0.5*\nlen+\xEone+0.5*\Ewidth,-0.5*\elllen) to (-0.5*\nlen+\xEone+0.5*\Ewidth,-\yArlen-0.5*\elllen);

\draw[<-,>=latex] (-0.5*\nlen+\xEtwo+0.5*\Ewidth,-0.5*\elllen) to (-0.5*\nlen+\xEtwo+0.5*\Ewidth,-\yArlen-0.5*\elllen);

\draw[<-,>=latex] (-0.5*\nlen+\xEthree+0.5*\Ewidth,-0.5*\elllen) to (-0.5*\nlen+\xEthree+0.5*\Ewidth,-\yArlen-0.5*\elllen);

\node[below] at (0.25,-\yArlen-0.5*\elllen) {error positions};

\draw[pattern=north west lines] (\xAdist-0.5*\tlen,-0.5*\elllen) rectangle (\xAdist+0.5*\tlen,0.5*\elllen);
\node[above] at (\xAdist+0,0.5*\elllen) {$t$};
\node at (\xAdist+0,-\ylabeldist) {$\A$};
\node at (\xAdist+0.5*\tlen+0.5*\myeps,0) {$\cdot$};

\draw (\xBdist-0.5*\nlen,-0.5*\tlen) rectangle (\xBdist+0.5*\nlen,0.5*\tlen);
\node[right] at (\xBdist+0.5*\nlen,0) {$t$};
\node[below] at (\xBdist+0,0.5*\elllen) {$n$};
\node at (\xBdist+0,-\ylabeldist) {$\B$};

\node at (\xBdist-0.5*\nlen+\xEone+0.5*\Ewidth,1/3*\tlen) {$1$};
\node at (\xBdist-0.5*\nlen+\xEtwo+0.5*\Ewidth,0/3*\tlen) {$1$};
\node at (\xBdist-0.5*\nlen+\xEthree+0.5*\Ewidth,-1/3*\tlen) {$1$};

\draw (\xBdist-0.5*\nlen+\xEone,-0.5*\tlen) rectangle (\xBdist-0.5*\nlen+\xEone+1*\Ewidth,0.5*\tlen);
\draw (\xBdist-0.5*\nlen+\xEtwo,-0.5*\tlen) rectangle (\xBdist-0.5*\nlen+\xEtwo+1*\Ewidth,0.5*\tlen);
\draw (\xBdist-0.5*\nlen+\xEthree,-0.5*\tlen) rectangle (\xBdist-0.5*\nlen+\xEthree+1*\Ewidth,0.5*\tlen);

\draw[<-,>=latex] (\xBdist-0.5*\nlen+\xEone+0.5*\Ewidth,-0.5*\tlen) to (\xBdist-0.5*\nlen+\xEone+0.5*\Ewidth,-\yArlen-0.5*\elllen);

\draw[<-,>=latex] (\xBdist-0.5*\nlen+\xEtwo+0.5*\Ewidth,-0.5*\tlen) to (\xBdist-0.5*\nlen+\xEtwo+0.5*\Ewidth,-\yArlen-0.5*\elllen);

\draw[<-,>=latex] (\xBdist-0.5*\nlen+\xEthree+0.5*\Ewidth,-0.5*\tlen) to (\xBdist-0.5*\nlen+\xEthree+0.5*\Ewidth,-\yArlen-0.5*\elllen);

\node[below] at (0.25+\xBdist,-\yArlen-0.5*\elllen) {error positions};

\node[left] at (\xAdist-0.5*\nlen-2*\myeps,-\Hsubdis) {$\Hsub$};
\node at (\xAdist-0.5*\nlen-\myeps,-\Hsubdis) {$=$};
\draw[pattern=north west lines] (\xAdist-0.5*\nlen,-\Hsubdis-0.5*\nktlen) rectangle (\xAdist+0.5*\nlen,-\Hsubdis+0.5*\nktlen);
\node[right] at (\xAdist+0.5*\nlen,-\Hsubdis) {$n-k-t$};
\node[above] at (\xAdist,0.5*\nktlen-\Hsubdis) {$n$};

\draw[fill=white] (\xAdist-0.5*\nlen+\xEone,-\Hsubdis-0.5*\nktlen) rectangle (\xAdist-0.5*\nlen+\xEone+1*\Ewidth,-\Hsubdis+0.5*\nktlen);

\draw[fill=white] (\xAdist-0.5*\nlen+\xEtwo,-\Hsubdis-0.5*\nktlen) rectangle (\xAdist-0.5*\nlen+\xEtwo+1*\Ewidth,-\Hsubdis+0.5*\nktlen);

\draw[fill=white] (\xAdist-0.5*\nlen+\xEthree,-\Hsubdis-0.5*\nktlen) rectangle (\xAdist-0.5*\nlen+\xEthree+1*\Ewidth,-\Hsubdis+0.5*\nktlen);

\draw[<-,>=latex] (\xAdist-0.5*\nlen+\xEone+0.5*\Ewidth,-0.5*\nktlen-\Hsubdis) to (\xAdist-0.5*\nlen+\xEone+0.5*\Ewidth,-\yArlenTwo-0.5*\elllen-\Hsubdis);

\draw[<-,>=latex] (\xAdist-0.5*\nlen+\xEtwo+0.5*\Ewidth,-0.5*\nktlen-\Hsubdis) to (\xAdist-0.5*\nlen+\xEtwo+0.5*\Ewidth,-\yArlenTwo-0.5*\elllen-\Hsubdis);

\draw[<-,>=latex] (\xAdist-0.5*\nlen+\xEthree+0.5*\Ewidth,-0.5*\nktlen-\Hsubdis) to (\xAdist-0.5*\nlen+\xEthree+0.5*\Ewidth,-\yArlenTwo-0.5*\elllen-\Hsubdis);

\node[below] at (\xAdist+0.25,-\yArlenTwo-0.5*\elllen-\Hsubdis) {zero columns exactly in error positions};

\node at (\xAdist,-\Formdist) {
$    \begin{aligned}
  \Rightarrow \mathrm{supp}_{\text{H}}(\E) = \bigcup_{i} \mathrm{supp}_{\text{H}}(\B_{i}) 
  = \overline{\bigcup_{i} \mathrm{supp}_{\text{H}}(\Hsubi)}
\end{aligned} $
    };

\draw (0.5*\rkdist+0.5*\xBdist,-10) to (0.5*\rkdist+0.5*\xBdist,3.5);


\draw[pattern=north west lines] (\rkdist-0.5*\nlen,-0.5*\elllen) rectangle (\rkdist+0.5*\nlen,0.5*\elllen);
\node[left] at (\rkdist-0.5*\nlen,0) {$\ell$};
\node[above] at (\rkdist,0.5*\elllen) {$n$};
\node at (\rkdist,-\ylabeldist) {$\E$};

\node at (\rkdist+0.5*\nlen+\myeps,0) {$=$};

\draw[pattern=north west lines] (\rkdist+\xAdist-0.5*\tlen,-0.5*\elllen) rectangle (\rkdist+\xAdist+0.5*\tlen,0.5*\elllen);
\node[above] at (\rkdist+\xAdist+0,0.5*\elllen) {$t$};
\node at (\rkdist+\xAdist,-\ylabeldist) {$\A$};
\node at (\rkdist+\xAdist+0.5*\tlen+0.5*\myeps,0) {$\cdot$};

\draw (\rkdist+\xBdist-0.5*\nlen,-0.5*\tlen) rectangle (\rkdist+\xBdist+0.5*\nlen,0.5*\tlen);
\node[right] at (\rkdist+\xBdist+0.5*\nlen,0) {$t$};
\node[below] at (\rkdist+\xBdist+0,0.5*\elllen) {$n$};
\node at (\rkdist+\xBdist+0,-\ylabeldist) {$\B$};
\node at (\rkdist+\xBdist,0) {$\Fq$};

\node[left] at (\rkdist+\xHsub-0.5*\nlen-1.25*\myeps,-\Hsubdis) {$\Hsub$};
\node at (\rkdist+\xHsub-0.5*\nlen-\myeps,-\Hsubdis) {$=$};
\draw (\rkdist+\xHsub-0.5*\nlen,-\Hsubdis-0.5*\nktlen) rectangle (\rkdist+\xHsub+0.5*\nlen,-\Hsubdis+0.5*\nktlen);

\draw (\rkdist+\xHsub-0.5*\nlen,-\Hsubdis+0.25*\nktlen) rectangle (\rkdist+\xHsub+0.5*\nlen,-\Hsubdis+0.5*\nktlen);
\node at (\rkdist+\xHsub,-\Hsubdis+0.375*\nktlen) {$\HsubEins$};

\node at (\rkdist+\xHsub,-\Hsubdis+0.15){$\vdots$};

\draw (\rkdist+\xHsub-0.5*\nlen,-\Hsubdis-0.5*\nktlen) rectangle (\rkdist+\xHsub+0.5*\nlen,-\Hsubdis-0.25*\nktlen);
\node at (\rkdist+\xHsub,-\Hsubdis-0.375*\nktlen) {$\HsubNKT$};

\node at (\rkdist+\xHsubExt-0.5*\nlen-3.5*\myeps,-\Hsubdis) {$\mapsto$};

\node at (\rkdist+\xHsubExt-0.5*\nlen,-\Hsubdis) {$\extsmallfield_{\boldsymbol{\gamma}}(\Hsub)$};

\node at (\rkdist+\xHsubExt-0.5*\nlen+0.75*\matExt,-\Hsubdis) {$=$};
\draw (\rkdist+\xHsubExt-0.5*\nlen+\matExt,-\Hsubdis-0.5*\extnktlen) rectangle (\rkdist+\xHsubExt+0.5*\nlen+\matExt,-\Hsubdis+0.5*\extnktlen);

\draw (\rkdist+\xHsubExt-0.5*\nlen+\matExt,-\Hsubdis+0.25*\extnktlen) rectangle (\rkdist+\xHsubExt+0.5*\nlen+\matExt,-\Hsubdis+0.5*\extnktlen);
\node[align=center] at (\rkdist+\xHsubExt+\matExt,-\Hsubdis+0.375*\extnktlen) {gen. set of \\$\ranksupp{\HsubEins}$};

\node at (\rkdist+\xHsubExt+\matExt,-\Hsubdis+0.15){$\vdots$};

\draw (\rkdist+\xHsubExt-0.5*\nlen+\matExt,-\Hsubdis-0.5*\extnktlen) rectangle (\rkdist+\xHsubExt+0.5*\nlen+\matExt,-\Hsubdis-0.25*\extnktlen);
\node[align=center] at (\rkdist+\xHsubExt+\matExt,-\Hsubdis-0.375*\extnktlen) {gen. set of \\$\ranksupp{\HsubNKT}$};

\node at (\rkdist+\xAdist,-\Formdist) {
$    \begin{aligned}
  \Rightarrow \ranksupp{\E}  = \sum_{i} \ranksupp{\B_{i}} 
   = \bigg(\sum_{i} \ranksupp{\Hsubi}\bigg)^{\bot}
\end{aligned} $
    };

\end{tikzpicture}
}
\end{center}
\caption{Illustration of the error support in Hamming metric (left part) and in rank metric (right part).}
\label{fig:support_illustration}
\end{figure}
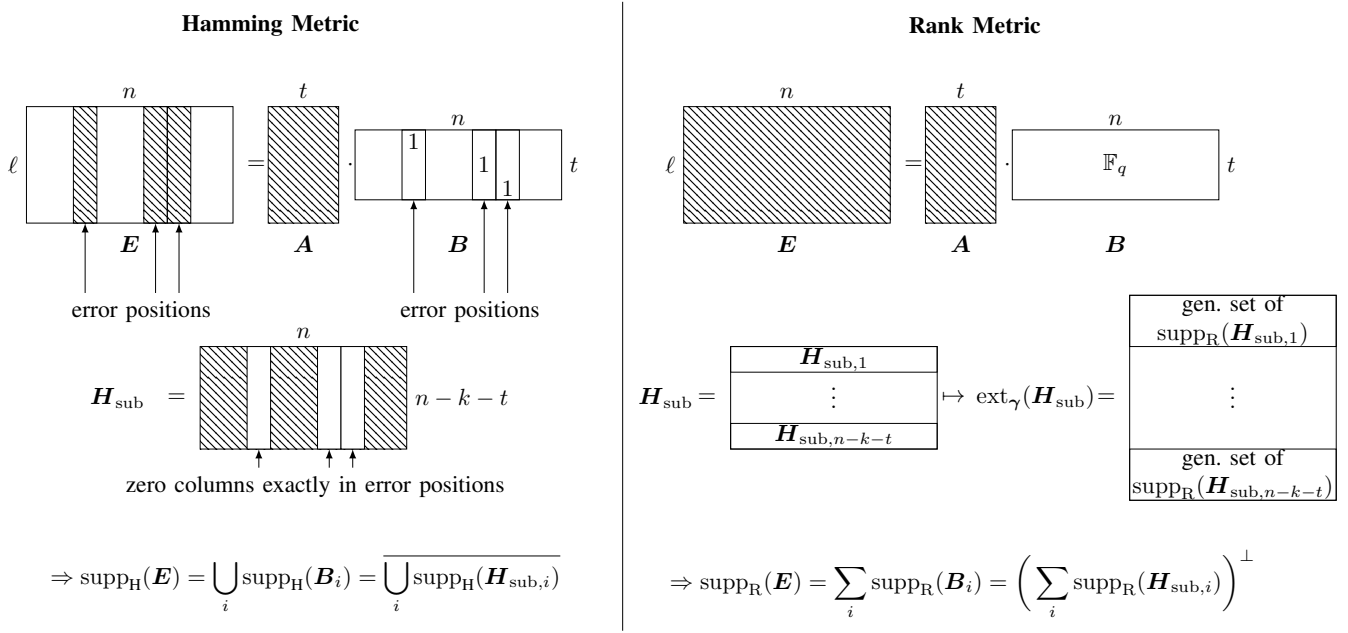

By \eqref{eq:Hamming_decomposition}, an error matrix $\E$ with $t$ errors in the Hamming metric can be decomposed into $\E = \A \B$, where the rows of $\B$ are the identity vectors corresponding to the $t$ error positions.
Hence, the support of the error matrix is given by the union of the supports of the rows $\B_i$ of the matrix $\B$, i.e.,
\begin{align*}
\mathrm{supp}_{\text{H}}(\E) = \bigcup_{i} \mathrm{supp}_{\text{H}}(\B_{i}).
\end{align*}
If the Metzner--Kapturowski algorithm works, the error positions are exactly the zero columns in the matrix $\Hsub$. The indices of the non-zero columns can be written as $\bigcup_{i} \mathrm{supp}_{\text{H}}(\Hsubi)$, where $\Hsubi$ is the $i$-th row of $\Hsub$.
Thus, the error positions are the complement ($=$ indices of the other rows) thereof, i.e.,
\begin{align}
\mathrm{supp}_{\text{H}}(\E)
= \overline{\bigcup_{i} \mathrm{supp}_{\text{H}}(\Hsubi)} \ . \label{eq:support_Hamming}
\end{align}

A very similar formula can be derived in the rank metric.
By Lemma~\ref{lem:E=AB}, an error in the rank metric can be decomposed into $\E=\A \B$, there $\B$ is a matrix over the small field $\Fq$ (note that this includes the Hamming metric case since $0,1 \in \Fqm$ are in the subfield $\Fq$).
The rank support of $\E$ equals the row space of $\extsmallfield(\B)$, which is spanned by the union (over $i$) of all rows of the $\extsmallfield(\B_i)$, where $\B_i$ is the $i$-th row of $\B$.
Since the row space of $\extsmallfield(\B_i)$ is the rank support of $\B_i$, the rank support of $\E$ is given by
\begin{align*}
\ranksupp{\E}  = \sum_{i} \ranksupp{\B_{i}},
\end{align*}
where the sum $\sum$ denotes the Minkowski sum of vector spaces (i.e., the span of the union of the spaces).

If the algorithm works correctly, the rank support of $\E$ is given by the kernel of $\extsmallfield(\Hsub)$ according to Theorem~\ref{thm:kernelB}.
This is the dual space of the row space of $\extsmallfield(\Hsub)$.
The row space of $\extsmallfield(\Hsub)$ is again the (Minkowski) sum of the row spaces of $\extsmallfield(\Hsubi)$, i.e., the rank support of the $i$-th row $\Hsubi$ of $\Hsub$.
Hence, we can write the rank support of $\E$ as
\begin{align}
\ranksupp{\E} 
= \bigg(\sum_{i} \ranksupp{\Hsubi}\bigg)^{\bot} \ . \label{eq:support_rank}
\end{align}

In summary, we see that there is an analogy of the notions of support in both metrics.
Equations \eqref{eq:support_Hamming} and \eqref{eq:support_rank} are dual to each other as outlined in Table~\ref{tab:duality_supports}.

\begin{table}[h]
\caption{Duality of Equations \eqref{eq:support_Hamming} and \eqref{eq:support_rank} (see also Figure~\ref{fig:support_illustration}).}
\label{tab:duality_supports}
\vspace{-0.3cm}
\begin{center}
\normalsize
\def\arraystretch{1.3}
\begin{tabular}{c|c}
Hamming metric & Rank metric \\
\hline
Support $\suppH$ & Rank support $\suppR$ \\
Union $\bigcup$ & Span of the union / Minkowski sum $\sum$ \\
Complement $\overline{\cdot}$ & Dual space $\cdot^\bot$
\end{tabular}
\end{center}
\vspace{-0.5cm}
\end{table}

\subsection{Example}
In this section, we illustrate the algorithm using example values. We consider the field $\mathbb{F}_{2^5}$ with the field polynomial $f(x) = x^5+x^2+1$ and primitive element $\alpha$. Further, we choose $\boldsymbol{\gamma} = \begin{mymatrix} 1 &\alpha &\alpha^2 & \alpha^3 & \alpha^4 \end{mymatrix}$ to be a polynomial basis and $\IntCode{\ell;n,k,d}$ as an interleaved Gabidulin code with $\ell=2$, $n=5$, $k=2$, $d=4$ and code locators $ 1,\alpha,\alpha^2,\alpha^3, \alpha^4$. We use the generator matrix
\begin{equation*}
  \G = \begin{mymatrix}
    1 & \alpha & \alpha^2 & \alpha^3 & \alpha^4\\
    1 & \alpha^2 & \alpha^4 & \alpha^6 & \alpha^8\\
    \end{mymatrix}
  \end{equation*}
 and parity-check matrix
  \begin{equation*}    
  \H = \begin{mymatrix}
    1 & 0 & 0 & \alpha^{17} & \alpha^{4}  \\
    0 & 1 & 0 & \alpha^{7}  & \alpha^{13} \\
    0 & 0 & 1 & \alpha^{16} & \alpha^{28} \\
    \end{mymatrix},
  \end{equation*}
  respectively.
  The transmitted codeword
  \begin{equation*}
  \C = \begin{mymatrix}
    \alpha & 1\\
    \alpha^2 & \alpha\\
    \end{mymatrix} \G = \begin{mymatrix}
    \alpha^{18} & 0 & \alpha^{21} & \alpha^{9}  & \alpha^{3} \\
    \alpha^{19} & 0 & \alpha^{22} & \alpha^{10} & \alpha^{4} \\
    \end{mymatrix}
  \end{equation*}
  is corrupted by an error
  \begin{equation*}
  \E = \begin{mymatrix}
    \alpha^3 & \alpha^1 & \alpha^3 & \alpha^1 & \alpha^1 \\
     \alpha^1 & \alpha^2 &  \alpha^1 & \alpha^2 & \alpha^2 \\
    \end{mymatrix},
  \end{equation*}  
  resulting in the received word
  \begin{equation*}
  \R = \C + \E = \begin{mymatrix}
    \alpha^{27} & \alpha^1 & \alpha^{4}   & \alpha^{21} & \alpha^{6} \\
    \alpha^{2}  & \alpha^2 & \alpha^{26}  & \alpha^{22} & \alpha^{7} \\
    \end{mymatrix}.
  \end{equation*}  
  To retrieve $\C$ from $\R$, we first compute the syndrome
  \begin{equation*}
  \S = \H  \R^{\top} = \begin{mymatrix}
    \alpha^{12}  & \alpha^{12}  \\
    \alpha^{30}  & 1 \\
    \alpha^{30}  & \alpha^{17}  \\
    \end{mymatrix}
  \end{equation*}  
  and bring it, using
  \begin{equation*}
  \P = \begin{mymatrix}
    0 & \alpha^{19} & \alpha^2  \\ 
    0 & \alpha    & \alpha    \\
    1 & \alpha^{14} & 1  \\
    \end{mymatrix},
  \end{equation*}  
  in row echelon form
  \begin{equation*}
  \P\S = \begin{mymatrix}
    1 & 0  \\
    0 & 1 \\
    0 & 0
    \end{mymatrix}.
  \end{equation*}  
  We take the last $n-k-t=1$ rows from 
  \begin{equation*}
  \P\H = \begin{mymatrix}
    0 & \alpha^{19} & \alpha^2 & \alpha^{7} & \alpha^{4} \\
    0 & \alpha & \alpha & \alpha^{24} & \alpha^{7}  \\
    1 & \alpha^{14} & 1 & \alpha^{4} & \alpha^{8} 
    \end{mymatrix},
  \end{equation*}    
  to get
  \begin{equation*}
  \Hsub = \begin{mymatrix}
    1 & \alpha^{14} & 1 & \alpha^{4} & \alpha^{8} \\
    \end{mymatrix}.
  \end{equation*}    
  Using the polynomial basis $\boldsymbol{\gamma}$, we extend $\Hsub$ to
  \begin{equation*}
  \extsmallfield(\Hsub) = \begin{mymatrix}
    1 & 1 & 1 & 0 & 1 \\
    0 & 0 & 0 & 0 & 0\\
    0 & 1 & 0 & 0 & 1\\
    0 & 1 & 0 & 0 & 1\\
    0 & 1 & 0 & 1 & 0\\
    \end{mymatrix}.
  \end{equation*}    

  The right $\Fq$-kernel of $\extsmallfield(\Hsub)$ is a two-dimensional subspace of $\Fq^n$ with basis vectors $\begin{mymatrix} 1 & 0 & 1 & 0 & 0 \end{mymatrix}$ and $\begin{mymatrix} 0 & 1 & 0 & 1 & 1 \end{mymatrix}$. Thus,
  \begin{equation*}
    \B = \begin{mymatrix}
    1 & 0 & 1 & 0 & 0 \\
    0 & 1 & 0 & 1 & 1
    \end{mymatrix}.
  \end{equation*}    
 Solving
  \begin{align*}
    \H\B^{\top}\A^{\top} &= \S \\
    \begin{mymatrix}
      1 & \alpha^{18} \\
      0 & \alpha^{29} \\
      1 & \alpha^{8}
    \end{mymatrix}\A^{\top} &=
\begin{mymatrix}
    \alpha^{12}  & \alpha^{12}  \\
    \alpha^{30}  & 1 \\
    \alpha^{30}  & \alpha^{17}  \\
    \end{mymatrix}
  \end{align*}    
  for $\A$ gives
\begin{equation*}
  \A = \begin{mymatrix}
    \alpha^{3} & \alpha  \\
    \alpha & \alpha^{2}
    \end{mymatrix}.
  \end{equation*}    
  The codeword is then retrieved by $\C = \R -  \A \B$.

\section{Further Results}\label{sec:further_results}

In this section, we present further results related to the new decoding algorithm.
We show that most errors of a given weight fulfill the full-rank condition, which directly implies a lower bound on the success probability for random errors, cf.~Section~\ref{ssec:number_matrices_lower_bound}.
By utilizing the well-known equivalence of homogeneous interleaved codes and certain linear codes over large extension fields, we obtain a class of linear rank-metric codes that are able to correct most rank errors of weight up to $d-2$, cf.~Section~\ref{ssec:linear_codes_correcting_d-2_errors}.
In Section~\ref{ssec:more_than_d-2}, we show that the new decoder is able to correct even more than the minimum distance number of errors in some cases.
We also give an idea how the algorithm can be adapted to some heterogeneous codes, cf.~Section~\ref{ssec:inhomogeneous_codes}.
Finally, we draw connections to existing decoding algorithms for the special case of interleaved Gabidulin codes in Section~\ref{ssec:relation_to_interleaved_gabidulin_decoders}: we show that these decoders succeed under the same condition as our new decoder, which significantly simplifies the known decoders' success conditions for the high-order case.

\subsection{A Bound on the Success Probability for Random Errors}\label{ssec:number_matrices_lower_bound}

In the following, we derive a lower bound on the success probability of the new decoder, given that the error is drawn uniformly at random from all error matrices of a given weight $\rank_{\Fq}(\E)=t$.
The bound is obtained by well-known results about the number of matrices of a given rank.
For fixed code parameters and $t\leq d-2$, the success probability of the new decoder approaches $1$ exponentially in the difference of the number of errors $t$ and the interleaving order $\ell$.

To derive the lower bound, we need the following lemma, which follows by well-known combinatorial results, see, e.g., \cite[Lemma~3.13]{Overbeck_Diss_InterleveadGab}.
\begin{lemma}\label{lem:lower_bound}
  Let $\ell\geq t$ and $\A$ be drawn uniformly at random from the set of matrices in $\Fqm^{\ell \times t}$ of $\Fq$-rank $\rank_{\Fq}(\A)=t$. Then, the probability that $\rank_{\Fqm}(\A) = t$ is lower-bounded by
  \begin{equation}
\geq \prod_{i=0}^{t-1} (1-q^{m(i-\ell)}) \geq 1-tq^{m(t-1-\ell)}. \label{eq:number_of_matrices_of_rank_t}
    \end{equation}
 \end{lemma}
 \begin{IEEEproof}
In the proof of \cite[Lemma~3.13]{Overbeck_Diss_InterleveadGab}, it is shown that the number of $n \times m$ matrices of rank $t$ over $\Fq$ is given by
\begin{equation*}
\NMtnm = \prod_{i=0}^{t-1} \frac{(q^m-q^i)(q^n-q^i)}{q^t-q^i}
\end{equation*}
which implies that the number of full $\Fqm$-rank matrices $\A$ is $\prod_{i=0}^{t-1}(q^{m\ell}-q^{mi})$.
Thus, a lower bound on the probability that that $\rank_{\Fqm}(\A) = t$ is given by\footnote{The first inequality is relatively tight due to the following argument. Equation \eqref{eq:number_of_matrices_of_rank_t} implies that the number of full $\Fq$-rank matrices $\A$ is $\prod_{i=0}^{t-1}(q^{m\ell}-q^{i})$. Since $m\ell\gg t$, the ratio of all full $\Fq$-rank matrices in $\Fqm^{\ell \times t}$ to all matrices in $\Fqm^{\ell \times t}$ is very close to $1$.}
\begin{align*}
&\frac{|\{ \A \in \Fqm^{\ell \times t} \, : \, \rank_{\Fqm}(\A) =t \}|}{|\{ \A \in \Fqm^{\ell \times t} \, : \, \rank_{\Fq}(\A) = t \}|} \\
&\quad \quad \geq \frac{|\{ \A \in \Fqm^{\ell \times t} \, : \, \rank_{\Fqm}(\A) = t \}|}{q^{m \ell t}} \\
&\quad \quad =  \frac{\prod_{i=0}^{t-1}(q^{m\ell}-q^{mi})}{q^{m\ell t}} \\
&\quad \quad = \prod_{i=0}^{t-1}\frac{(q^{m\ell}-q^{mi})}{q^{m\ell t}} \\
&\quad \quad = \prod_{i=0}^{t-1}(1-q^{m(t-1-\ell)}) \\
&\quad \quad \geq (1-q^{m(t-1-\ell)})^t \\
&\quad \quad = \sum_{i=0}^{t} \binom{t}{i} (-q^{m(t-1-\ell)})^{i} \\
&\quad \quad = 1 - \binom{t}{1} q^{m(t-1-\ell)} +\binom{t}{2} q^{2m(t-1-\ell)} - \hdots \\
&\quad \quad \geq 1 - tq^{m(t-1-\ell)}.
  \end{align*}
  The latter inequality results from 
\begin{align*}
  \binom{t}{i} < \left(\frac{t e}{i}\right)^i \leq (t e)^i &\leq ((d-2)e)^i \leq ((m-2)e)^i \\
                                                           &< q^{im}\leq q^{-im(t-1-\ell)},
\end{align*}
for $q\geq 2$ and $i\geq 1$, where $e$ denotes Euler's number.
Thus,
\begin{equation*}
\binom{t}{i} q^{im(t-1-\ell)}  > \binom{t}{(i+1)} q^{(i+1)m(t-1-\ell)}.
\end{equation*}
 \end{IEEEproof}

 The lower bound on the success probability of the proposed algorithm is then given by the following theorem.
\begin{theorem}\label{thm:success_probability}
Let $\Code$ be a homogeneous $\ell$-interlaved code of minimum rank distance $d$, and $t \leq \min\{\ell, d-2\}$.
Furthermore, let
\begin{align*}
\R = \C + \E,
\end{align*}
where $\C$ is a codeword of $\Code$ and $\E$ is uniformly drawn from $\Fqm^{\ell \times n}$ with $\rank_{\Fq}(\E)=t$.
Then, with probability
\begin{equation*}
\geq  \prod_{i=0}^{t-1} (1-q^{m(i-\ell)}) \geq  1- tq^{m(t-1-\ell)},
  \end{equation*}
$\C$ can be uniquely reconstructed from $\R$ using Algorithm~\ref{alg:Decoding}.
\end{theorem}
\begin{IEEEproof}
According to Theorem~\ref{thm:main_statement}, one can uniquely reconstruct $\C$ from $\R$ using Algorithm~\ref{alg:Decoding} if $\rank_{\Fqm}(\E)=t$.
Considering the decomposition of $\E$ shown in Lemma~\ref{lem:E=AB}, the error $\E$ has $\Fqm$-rank $t$ if and only if the matrix $\A \in \Fqm^{\ell \times t}$ has $\Fqm$-rank $t$.
Furthermore, for a given subspace $\mathcal{V} \subseteq \Fq^n$ of $\Fq$-dimension $t$ and a fixed basis $\B$ whose row space equals $\mathcal{V}$, there is a bijective mapping of the set of matrices $\A \in \Fqm^{\ell \times t}$ of $\Fqm$-rank $t$ and the set of errors $\E \in \Fqm^{\ell \times n}$ of rank support $\mathcal{V}$, given by $\A \mapsto \A \B$.
Hence, drawing $\E$ uniformly at random with rank weight $t$ is equivalent to drawing $\B \in \Fq^{t \times n}$ and $\A \in \Fqm^{\ell \times t}$, both of $\Fq$-rank $t$, uniformly at random.
Thus, the probability that $\E$ has $\Fqm$-rank $t$ is lower-bounded by
\begin{equation*}
\prod_{i=0}^{t-1} (1-q^{m(i-\ell)}) \geq 1-tq^{m(t-1-\ell)},
\end{equation*}
according to Lemma~\ref{lem:lower_bound}.
This implies that $\C$ can be uniquely reconstructed from $\R$ by Algorithm~\ref{alg:Decoding} with this probability.
\end{IEEEproof}

\subsection{Linear Rank-Metric Codes Correcting $d-2$ Errors Probabilistically}\label{ssec:linear_codes_correcting_d-2_errors}

A homogeneous $\ell$-interleaved code over the field $\Fqm$ with $\Fqm$-linear constituent code $\Code$ can be interpreted as an $\Fqmell$-linear rank-metric code by representing the columns of the codeword matrix, which are vectors in $\Fqm^\ell$, as elements of $\Fqmell$.
For interleaved Reed--Solomon codes in the Hamming metric, this was first proved in \cite{sidorenko2008decoding}, where one is even able to show that the resulting code is a punctured Reed--Solomon code over a large field.
The extension of this result to rank-metric codes is straightforward and, though to our knowledge unpublished, appears to be common knowledge in the community.
The following statement formally states the result for general linear rank-metric codes.

\begin{lemma}\label{lem:interleaved_code}
Let $\Code[n,k,d]$ be a rank-metric code over $\Fqm$. Then, the code $\Code'$ obtained by representing the codewords of a homogeneous $\ell$-interleaved code thereof as vectors over the large field $\mathbb{F}_{q^{m \ell}}$, i.e.,
\begin{align*}
\Code' := \left\{ \extsmallfield^{-1}\left( \begin{bmatrix}
\c_1 \\
\c_2 \\
\vdots \\
\c_\ell
\end{bmatrix} \right)  \, : \, \c_i \in \Code \right\} \subseteq \Fqmell^n,
\end{align*}
is an $\Fqmell$-linear $[n,k,d]$ rank-metric code over $\Fqmell$.
\end{lemma}

\begin{proof}
Due to the linearity of the ext mapping, the sum of two codewords is again a codeword.
Furthermore, $\Fqmell$-multiples of a codeword result again in a codeword since multiplication by an $\Fqmell$-element can be seen as a multiplication with the corresponding matrix representation of the field element from the left to the interleaved codeword (seen as an $\ell \times n$ matrix over $\Fqm$).
Hence, only row operations are performed and the result is again a codeword since each row contains a linear combination of the previous rows, which are again codewords of $\Code$ due to the linearity of $\Code$.
The parameters follow directly by counting.
\end{proof}

Note that the overall extension degree $m\ell$ of the resulting code is typically much larger than the code length $n$.
On the other hand, for high interleaving order, the code is able to correct most error patterns up to $d-2$ instead of the usual $\tfrac{d-1}{2}$ errors.
We formalize this fact in the following statement, which implies that a random rank error of weight $t\leq d-2$ can be corrected with probability approaching $1$ and exponentially in the difference of $t$ and $\ell$.

\begin{theorem}
Let $\Code[n,k,d]$ be a linear rank-metric code over $\Fqm$ and $\Code'[n,k,d]$ be the corresponding linear code over $\Fqmell$ as constructed in Lemma~\ref{lem:interleaved_code}, where $\ell \in \mathbb{N}$. If $\ell \geq t$, we can can correct a fraction at least
\begin{align*}
\prod_{i=0}^{t-1} (1-q^{m(i-\ell)}) \geq  1- tq^{m(t-1-\ell)},
\end{align*}
of all rank errors $\e \in \Fqmell^n$ of rank weight $t$ with $\Code'$.
\end{theorem}

\begin{proof}
The statement directly follows by extending the received word into an $\ell \times n$ matrix over $\Fqm$, which is an interleaved codeword plus an error by Lemma~\ref{lem:interleaved_code}.
Furthermore, the mapping $\extsmallfield \, : \, \Fqmell^n \to \Fqm^{\ell \times n}$ is a bijective mapping between all errors $\e \in \Fqmell^n$ of rank weight $\rank_{\Fq}(\e)=t$ and interleaved error matrices of weight $t$.
Hence, we can correct the given fraction of errors using the interleaved decoder, Algorithm~\ref{alg:Decoding}, by Theorem~\ref{thm:success_probability}.
\end{proof}

\subsection{Correcting More Than $d-2$ Errors}\label{ssec:more_than_d-2}

In Theorem~\ref{thm:kernelB}, we showed that the rank support of an error $\E$ can be obtained from a received word $\C+\E$ if three conditions are fulfilled:
\begin{enumerate}[label=\roman*)]
\item $t \leq d-2$, \label{itm:t_<=_d-2}
\item $\ell \geq t$ (high-order condition), and
\item $\rank_{\Fqm} (\E) = t$ (full-rank condition).
\end{enumerate}
However, by the proof of Theorem~\ref{thm:kernelB}, it suffices to replace the first condition, \ref{itm:t_<=_d-2}, by
\begin{equation}
\rank_{\Fqm}\left(\H \left[\B^\top \mid \b^\top\right]\right) = t+1 \quad \forall \, \b \in \Fq^n \setminus \rspace{\B}{\Fq}, \label{eq:alternative_condition}
\end{equation}
where $\B$ is any basis of the rank support of $\E$.
Obviously, this condition is fulfilled for $t \leq d-2$ due to \cite[Theorem~1]{Gabidulin_TheoryOfCodes_1985}, but might as well be true for larger values of $t$, depending on the rank support.
Contrarily, by definition of the minimum distance, there must always be an error of weight $t > d-2$ for which \eqref{eq:alternative_condition} is not fulfilled.
Furthermore, it can only work up to $t+1 \leq n-k$ since $\H$ has only $n-k$ columns.

This condition is equivalent to saying that any error $\E'$ of weight $t+1$ whose rank support contains the one of $\E$, i.e., $\ranksupp{\E} \subseteq \ranksupp{\E'}$, can be uniquely reconstructed from $\ranksupp{\E'}$ using the erasure-correction strategy described in Lemma~\ref{lem:get_E_from_B}.

In simulations, we observed that decoding more than $d-2$ errors can succeed even with high probability. E.g., in case of a linear $[n=10,k=2,d=7]$ code over $\mathbb{F}_{2^{10}}$ and an error with $t =\rank_{\Fqm}(\E) = \rank_{\Fq}(\E) = n-k-1=7$, we obtained a success probability of more than $99\%$.

There is again an analogy to the Hamming metric: The Metzner--Kapturowski algorithm works if and only if the $t+1$ columns of the parity-check matrix corresponding to the error positions (support) and any one additional column have full rank.
Hence, any error consisting of the original $t$ error positions and one additional one must be uniquely reconstructable from the syndrome and known support in order for the algorithm to work.
This fact was recently utilized for decoding interleaved locally repairable codes (LRCs) (in the Hamming metric) far beyond their minimum distance \cite{holzbaur2019lrc}. A family of LRCs was derived, which fulfills this condition for the majority of $t$ error positions, where $t$ is greater than the minimum distance.

\subsection{Heterogeneous Interleaved Codes}\label{ssec:inhomogeneous_codes}

The fact that all constituent codes of the interleaved code must have the same parity-check matrix restricts the new decoder, Algorithm~\ref{alg:Decoding}, to homogeneous interleaved codes.
In the following, we show that it also works for certain heterogeneous
codes.
An heterogeneous interleaved code is defined by
\begin{align*}
\IC_\mathrm{het}[\ell; n,k_1,\dots,k_\ell] := \left\{ \C = \begin{bmatrix}
\c_1 \\ \c_2 \\ \vdots \\ \c_\ell
\end{bmatrix} \, : \, \c_i \in \Code_i[n,k_i,d_i] \right\} ,
\end{align*}
where $\Code[n,k_i,d_i]$ are linear rank-metric codes. Note that in contrast to Definition~\ref{def:interleaved_codes}, the constituent codes $\Code_i$ can be distinct and also have different parameters.

If all the constituent codes are contained in a joint supercode $\Code_i \subseteq \Code[n,k,d]$ (e.g., $\Code$ could be one of the constituent codes, which would imply $d = \min_i\{d_i\}$), then obviously we have
\begin{equation*}
\IC_\mathrm{het}[\ell; n,k_1,\dots,k_\ell] \subseteq \IC[\ell; n,k,d],
\end{equation*}
where $\IC$ is the homogeneous $\ell$-interleaved code of $\Code$.
Hence, we can correct up to $d-2$ errors in $\IC[\ell; n,k_1,\dots,k_\ell]$ using a parity-check matrix of $\Code$.

Consider again the general case.
Let $j \in \{1,\dots,\ell\}$ be such that $d_j = d := \min_i \{d_i\}$.
If for more than $d-1$ errors, the rank support of the error is found correctly, the $j$-th row of the error matrix $\E$ is not always uniquely determined by $\B$.
This corresponds to the erasure correction capability of the $j$-th constituent code.
In particular, the error is not unique if and only if the rank of $\H \B$ is not equal to the number of errors $t$.
If $\Code$ is an MRD code, there is never a unique solution.

In summary, there are heterogeneous rank-metric codes that correct up to
\begin{equation*}
\min_i\{d_i\}-2 \text{ errors}
\end{equation*}
under the high-order and full-rank condition. Contrarily, it is impossible to always uniquely reconstruct more than
\begin{equation*}
\min_i\{d_i\}-1 \text{ errors}
\end{equation*}
even if the rank support is known.
This shows that the new decoder is close to optimal (in this sense) for those heterogeneous codes whose constituent codes are subcodes of one of its constituent codes.

\subsection{Relation to Existing Decoders of Interleaved Codes}\label{ssec:relation_to_interleaved_gabidulin_decoders}

There are several known algorithms for decoding interleaved Gabidulin codes, i.e., the special case in which the constituent codes are Gabidulin codes.
The first one was proposed by Loidreau and Overbeck in \cite{loidreau2006decoding} (see also \cite{Overbeck_Diss_InterleveadGab}) and is based on solving a linear system of equations.
Later, Sidorenko and Bossert \cite{sidorenko2010decoding} proposed a syndrome key equation based algorithm for a slightly different error model (sometimes referred to as \emph{horizontally interleaved codes}).
It was shown in \cite{wachter2013decoding} that the Sidorenko--Bossert algorithm can be adapted to the error model assumed here.
Furthermore, there is an interpolation-based decoding algorithm \cite{wachter2014list}.

All three decoding algorithms\footnote{Alternatively, the Wachter-Zeh--Zeh algorithm can be seen as a list decoder with exponential worst-case, but small average-case list size.} can be seen as partial (also called probabilistic) unique decoding algorithms, which are able to correct errors up to
\begin{equation}
t \leq \tmax := \left\lfloor \frac{\ell}{\ell+1}(d-1) \right\rfloor = \left\lfloor \frac{\ell}{\ell+1}(n-k) \right\rfloor \label{eq:IGab_radii}
\end{equation}
errors, but fail for some errors beyond half the minimum distance. For all three algorithms, there exist sufficient success conditions based on the structure of the error, which in general differ from the full-rank condition of our algorithm. There are also lower bounds on the success probability for all three algorithms, which are derived from these conditions.

As for sufficiently large interleaving order $\ell \geq d-2$ (high-order condition), the decoding radii of the interleaved Gabidulin decoders, $\tmax = d-2$, coincide with our new decoder, it is an interesting question to study their relation.

The following theorem proves that all three known decoders succeed under the high-order and full-rank condition.
For high-order interleaved Gabidulin codes, this yields a much easier success condition on the error than the known ones.
Although this means that our decoder works as well as the existing decoders for interleaved Gabidulin codes, our decoder is more general as it can be used for arbitrary linear rank-metric codes.
Note that a similar result was recently obtained in \cite{elleuch2018interleaved} for the Metzner--Kapturowski algorithm in the Hamming metric and interleaved Reed--Solomon codes.

\begin{theorem}\label{thm:success_known_IGab_decoders}
Let $\C$ be a codeword of an $\ell$-interleaved Gabidulin code of parameters $[n,k,n-k+1]$ and $\E \in \Fqm^{\ell \times n}$ be an error matrix with $\rank_{\Fq}(\E) = t$ that fulfills $t\leq \ell$ (\textbf{high-order condition}) and $\rank_{\Fqm}(\E) = t$ (\textbf{full-rank condition}).

Then, $\C$ can be uniquely reconstructed from the received word $\R = \C+\E$ using the Loidreau--Overbeck \cite{loidreau2006decoding}, Sidorenko--Bossert \cite{sidorenko2010decoding}\footnote{In order to work with the error model assumed here, we must use the variant described in \cite[Section~4.1]{wachter2013decoding}}, or Wachter-Zeh--Zeh \cite{wachter2014list} decoder.
  \end{theorem}
  \begin{IEEEproof}
	We need to show that all three decoders succeed under the given conditions.
	It suffices to prove that the Loidreau--Overbeck decoder does not fail since this implies the success of the other two algorithms (cf.~\cite[Lemma~4.1]{wachter2013decoding} and \cite[Lemma~4.8]{wachter2013decoding}).

    The algorithm in~\cite{loidreau2006decoding} reconstructs $\C$ from $\R$ uniquely if and only if (recall that $[i] := q^i$)
    \begin{equation}
     \rank_{\Fqm} \left(
 \begin{mymatrix}
        \g\\
        \g^{[1]}\\
        \vdots \\
        \g^{[n-t-2]}\\
        \E \\
        \E^{[1]} \\
        \vdots \\
        \E^{[n-k-t-1]} \\
      \end{mymatrix}       \right)
    = n-1.
    \label{eq:condition-loid}
      \end{equation}
      In the following, we show that the submatrix
      \begin{equation*}
        \tilde{\G} := 
 \begin{mymatrix}
        \g\\
        \g^{[1]}\\
        \vdots \\
        \g^{[n-t-2]}\\
      \end{mymatrix}
    \end{equation*}
has $\Fqm$-rank $n-t-1$ and that the submatrix
      \begin{equation*}
        \Z := 
       \begin{mymatrix}
        \E \\
        \E^{[1]} \\
        \vdots \\
        \E^{[n-k-t-1]} \\
      \end{mymatrix}
    \end{equation*}
    has $\Fqm$-rank $t$. Further, we prove that the rows of $\tilde{\G}$ and the rows of $\Z$ are linearly independent such that their ranks add up to $n-1$ and thus equation~(\ref{eq:condition-loid}) is fulfilled.

 The matrix $\tilde{\G}$ is a generator matrix of an $[n,n-t-1,t+2]$ Gabidulin code. Thus, $\tilde{\G}$ has $\Fqm$-rank $n-t-1$ and each $\Fqm$-linear combination of $\g,\hdots,\g^{[n-t-2]}$ has rank weight of at least $t+2$.

 Since $\rank_{\Fqm}(\E) = t$, it follows that $\rank_{\Fqm}(\Z) \geq t$. Further since $\rank_{\Fq}(\E)=t$, there exists a full-rank matrix $\P \in \Fq^{n \times n}$ such that
    \begin{equation*}
      \E\P = \begin{mymatrix}
\tilde{\E} & \0_{\ell,n-t}
        \end{mymatrix},
      \end{equation*}
      where $\tilde{\E} \in \Fqm^{\ell \times t}$. Thus,
      \begin{align*}
        \rank_{\Fqm} (\Z) &= \rank_{\Fqm} (\Z \P) \\
                          &=
        \rank_{\Fqm} \left(
       \begin{mymatrix}
        \tilde{\E} &\0_{\ell,n-t} \\
        \tilde{\E}^{[1]}&\0_{\ell,n-t} \\
        \vdots & \vdots\\
        \tilde{\E}^{[n-k-t-1]}&\0_{\ell,n-t} \\
      \end{mymatrix} \right) \\
        & \leq t,
          \end{align*}
since $\tilde{\E}$ has $t$ rows. This implies that $\rank_{\Fqm}(\Z) =t$ and each $\Fqm$-linear combination has rank weight of at most $t$.

  Since each $\Fqm$-linear combination of the rows of $\tilde{\G}$ has rank weight at least $t+2$ and each linear $\Fqm$-linear combination of the rows of $\Z$ has rank weight at most $t$, the rows of $\tilde{\G}$ and the rows of $\Z$ must be linearly independent.
    \end{IEEEproof}

\begin{remark}
The fastest variants of the known decoders for interleaved Gabidulin codes, \cite{sidorenko2011skew,sidorenko2014fast,puchinger2017row,puchinger2017alekhnovich}, have complexities $O(\ell n^2)$ over $\Fqm$ (using \cite{sidorenko2014fast} or \cite{puchinger2017row}), $O^\sim(\ell^3 n^{1.635})$ over $\Fqm$ (using \cite{sidorenko2011skew} or \cite{puchinger2017alekhnovich} with the fast linearized polynomial operations of \cite{puchinger2018fast}), and $O(\ell^3 n^{\omega-2} m^2)$ over $\Fq$ (using \cite{sidorenko2011skew} or \cite{puchinger2017alekhnovich} with the fast linearized polynomial operations of \cite{caruso2017fast}), where $2 \leq \omega \leq 3$ is the matrix multiplication exponent.

Hence, for high-order interleaved Gabidulin codes ($\ell \geq t$) and a reasonably large number of errors $t \in \Omega(n)$, the existing algorithms are not asymptotically faster than the new decoder.
This is surprising since the existing algorithms rely on the heavy structure of Gabidulin codes, whereas the new decoder is generic.
\end{remark}

\section{Conclusion}\label{sec:conclusion}

\subsection{Summary}

We have presented a new decoding algorithm for homogeneous interleaved codes of minimum rank distance $d$, which corrects up to $d-2$ errors given that the error matrix fulfills two conditions: the number of errors is at most the interleaving order (high-order condition) and the rank of the error, over the large field, equals the number of errors (full-rank condition). The new decoder is an adaption of Metzner and Kapturowski's decoder in the Hamming metric, works for any rank-metric code, consists only of linear-algebraic operations, and can be implemented with asymptotically
\begin{equation*}
O^\sim\!\left( \max\{n^2\ell,n^3\} m \right)
\end{equation*}
operations in the small field $\Fq$.

We have shown that the success probability for a random error of a given weight can be arbitrarily close to $1$ for growing interleaving order.
By viewing a homogeneous interleaved code as a linear code over a large extension field, we have observed that one obtains a linear rank-metric code (over a large field) correcting almost any error of rank weight up to $d-2$.
We have derived sufficient conditions on the rank support of the error for which the proposed algorithm is able to correct more than $d-2$ errors and presented a way of adapting it to certain heterogeneous codes.
In the special case of interleaved Gabidulin codes, we have proven that the known decoding algorithms succeed under the same conditions as the new decoder, which simplifies the existing success conditions.

\subsection{Open Problems}
The algorithm proposed in this paper is designed for interleaved rank-metric codes in which the constituent codewords and error vectors are aligned vertically (also called vertically interleaved codes). This error model was considered in the majority of works on interleaved rank-metric codes, including \cite{loidreau2006decoding,wachter2014list,wachter2013decoding}.
There is an alternative scheme, proposed in \cite{sidorenko2010decoding}, where where the constituent codewords are aligned horizontally.
In this case, the constituent error matrices' column spaces are all contained in a relatively small joint column space of dimension $t$.
An open problem for future research is the adaption of the new algorithm to this model.

In \cite{haslach2000efficient}, Haslach and Vinck generalized the algorithm of Metzner and Kapturowski to linearly dependent (Hamming-)error matrices. An interesting open question is whether our algorithm can also be generalized to such errors in the rank metric.

A further open problem relates to McEliece-like \cite{mceliece1978public} code-based cryptosystems in the rank metric (e.g., \cite{gabidulin1991ideals,gaborit2013low,loidreau2016evolution}), which are based on the hardness of decoding in a generic rank-metric code.
To encrypt a secret message, the message is mapped to a codeword of a code and the codeword is corrupted by an artificial error. An authorized user knows secret properties about the code that enables him or her to efficiently decode the corrupted codeword and thus retrieve the secret message. To an unauthorized user, the encrypted message looks like a corrupted codeword of a random code and is thus very complex to decode.
It might occur that, in different encryption rounds, two or more error matrices' row spaces are contained in a low-dimensional joint row space.
This could, e.g., occur by chance, a bad implementation of the cryptosystem, or---in a controlled way---on purpose (consider a rank-metric variant of \cite{elleuch2018interleaved,holzbaur2019decoding}).
Since the vertical alignment of the received words would result in an interleaved codeword of a seemingly random constituent code and a low-weight error matrix, the security level of the system might be affected by the new decoder (recall that it works for arbitrary codes).
However, the analysis is not straight-forward since one needs to study how the output of the new algorithm can be used to decrease the search space of a generic decoder (i.e., for realistic, low, interleaving orders, where the algorithm does not directly succeed).
We therefore leave this analysis as an open problem.

\section*{Acknowledgment}

We would like to thank Vladimir Sidorenko for making us aware of the Metzner--Kapturowski algorithm.

\bibliographystyle{IEEEtran}
\bibliography{main}

\begin{thebibliography}{10}
\providecommand{\url}[1]{#1}
\csname url@samestyle\endcsname
\providecommand{\newblock}{\relax}
\providecommand{\bibinfo}[2]{#2}
\providecommand{\BIBentrySTDinterwordspacing}{\spaceskip=0pt\relax}
\providecommand{\BIBentryALTinterwordstretchfactor}{4}
\providecommand{\BIBentryALTinterwordspacing}{\spaceskip=\fontdimen2\font plus
\BIBentryALTinterwordstretchfactor\fontdimen3\font minus
  \fontdimen4\font\relax}
\providecommand{\BIBforeignlanguage}[2]{{%
\expandafter\ifx\csname l@#1\endcsname\relax
\typeout{** WARNING: IEEEtran.bst: No hyphenation pattern has been}%
\typeout{** loaded for the language `#1'. Using the pattern for}%
\typeout{** the default language instead.}%
\else
\language=\csname l@#1\endcsname
\fi
#2}}
\providecommand{\BIBdecl}{\relax}
\BIBdecl

\bibitem{metzner1990general}
J.~J. Metzner and E.~J. Kapturowski, ``{A General Decoding Technique Applicable
  to Replicated File Disagreement Location and Concatenated Code Decoding},''
  \emph{IEEE Transactions on Information Theory}, vol.~36, no.~4, pp. 911--917,
  1990.

\bibitem{krachkovsky1997decoding}
V.~Y. Krachkovsky and Y.~X. Lee, ``{Decoding for Iterative Reed--Solomon Coding
  Schemes},'' \emph{IEEE Transactions on Magnetics}, vol.~33, no.~5, pp.
  2740--2742, 1997.

\bibitem{krachkovsky1998decoding}
------, ``{Decoding of Parallel Reed--Solomon Codes with Applications to
  Product and Concatenated Codes},'' in \emph{IEEE International Symposium on
  Information Theory}, 1998, p.~55.

\bibitem{haslach1999decoding}
C.~Haslach and A.~H. Vinck, ``{A Decoding Algorithm With Restrictions for Array
  Codes},'' \emph{IEEE Transactions on Information Theory}, vol.~45, no.~7, pp.
  2339--2344, 1999.

\bibitem{justesen2004decoding}
J.~Justesen, C.~Thommesen, and T.~H{\o}holdt, ``{Decoding of Concatenated Codes
  with Interleaved Outer Codes},'' in \emph{IEEE International Symposium on
  Information Theory (ISIT)}, 2004, pp. 328--328.

\bibitem{schmidt2005interleaved}
G.~Schmidt, V.~R. Sidorenko, and M.~Bossert, ``{Interleaved Reed--Solomon Codes
  in Concatenated Code Designs},'' in \emph{IEEE Information Theory Workshop},
  2005, pp. 5--pp.

\bibitem{schmidt2009collaborative}
------, ``{Collaborative Decoding of Interleaved Reed--Solomon Codes and
  Concatenated Code Designs},'' \emph{IEEE Transactions on Information Theory},
  vol.~55, no.~7, pp. 2991--3012, 2009.

\bibitem{schmidt2010syndrome}
------, ``{Syndrome Decoding of Reed--Solomon Codes Beyond Half the Minimum
  Distance Based on Shift-Register Synthesis},'' \emph{IEEE Transactions on
  Information Theory}, vol.~56, no.~10, pp. 5245--5252, 2010.

\bibitem{kampf2014bounds}
S.~Kampf, ``{Bounds on Collaborative Decoding of Interleaved Hermitian Codes
  and Virtual Extension},'' \emph{Designs, Codes and Cryptography}, vol.~70,
  no. 1-2, pp. 9--25, 2014.

\bibitem{rosenkilde2018power}
J.~Rosenkilde, ``{Power Decoding Reed--Solomon Codes up to the Johnson
  Radius},'' \emph{Advances in Mathematics of Communications}, vol.~12, no.~1,
  pp. 81--106, 2018.

\bibitem{puchinger2019improved}
S.~Puchinger, J.~Rosenkilde, and I.~Bouw, ``{Improved Power Decoding of
  Interleaved One-Point Hermitian Codes},'' \emph{Designs, Codes and
  Cryptography}, vol.~87, no. 2-3, pp. 589--607, 2019.

\bibitem{elleuch2018interleaved}
\BIBentryALTinterwordspacing
M.~Elleuch, A.~{Wachter-Zeh}, and A.~Zeh, ``{A Public-Key Cryptosystem from
  Interleaved Goppa Codes},'' 2018. [Online]. Available:
  \url{http://arxiv.org/abs/1809.03024}
\BIBentrySTDinterwordspacing

\bibitem{holzbaur2019decoding}
L.~Holzbaur, H.~Liu, S.~Puchinger, and A.~Wachter-Zeh, ``{On Decoding and
  Applications of Interleaved Goppa Codes},'' in \emph{accepted at: IEEE
  International Symposium on Information Theory (ISIT)}, 2019.

\bibitem{bleichenbacher2003decoding}
D.~Bleichenbacher, A.~Kiayias, and M.~Yung, ``{Decoding of Interleaved Reed
  Solomon Codes Over Noisy Data},'' in \emph{International Colloquium on
  Automata, Languages, and Programming}.\hskip 1em plus 0.5em minus 0.4em\relax
  Springer, 2003, pp. 97--108.

\bibitem{coppersmith2003reconstructing}
D.~Coppersmith and M.~Sudan, ``{Reconstructing Curves in Three (and Higher)
  Dimensional Space from Noisy Data},'' in \emph{ACM Symposium on the Theory of
  Computing}, 2003.

\bibitem{parvaresh2004multivariate}
F.~Parvaresh and A.~Vardy, ``{Multivariate Interpolation Decoding Beyond the
  Guruswami--Sudan Radius},'' in \emph{Allerton Conference on Communication,
  Control and Computing}, 2004.

\bibitem{brown2004probabilistic}
A.~Brown, L.~Minder, and A.~Shokrollahi, ``{Probabilistic Decoding of
  Interleaved RS-Codes on the q-Ary Symmetric Channel},'' in \emph{IEEE
  International Symposium on Information Theory (ISIT)}, 2004, pp. 326--326.

\bibitem{parvaresh2007algebraic}
F.~Parvaresh, ``{Algebraic List-Decoding of Error-Correcting Codes},'' Ph.D.
  dissertation, University of California, San Diego, 2007.

\bibitem{schmidt2007enhancing}
G.~Schmidt, V.~Sidorenko, and M.~Bossert, ``{Enhancing the Correcting Radius of
  Interleaved Reed--Solomon Decoding Using Syndrome Extension Techniques},'' in
  \emph{IEEE International Symposium on Information Theory (ISIT)}, 2007, pp.
  1341--1345.

\bibitem{cohn2013approximate}
H.~Cohn and N.~Heninger, ``{Approximate Common Divisors via Lattices},''
  \emph{The Open Book Series}, vol.~1, no.~1, pp. 271--293, 2013.

\bibitem{nielsen2013generalised}
J.~S. Nielsen, ``{Generalised Multi-Sequence Shift-Register Synthesis Using
  Module Minimisation},'' in \emph{IEEE International Symposium on Information
  Theory (ISIT)}, 2013, pp. 882--886.

\bibitem{wachterzeh2014decoding}
A.~{Wachter-Zeh}, A.~Zeh, and M.~Bossert, ``{Decoding Interleaved Reed--Solomon
  Codes Beyond Their Joint Error-Correcting Capability},'' \emph{Designs, Codes
  and Cryptography}, vol.~71, no.~2, pp. 261--281, 2014.

\bibitem{puchinger2017irs}
S.~Puchinger and J.~{Rosenkilde n\'e Nielsen}, ``{Decoding of Interleaved
  Reed--Solomon Codes Using Improved Power Decoding},'' in \emph{IEEE
  International Symposium on Information Theory (ISIT)}, 2017.

\bibitem{yu2018simultaneous}
J.-H. Yu and H.-A. Loeliger, ``{Simultaneous Partial Inverses and Decoding
  Interleaved Reed--Solomon Codes},'' \emph{IEEE Transactions on Information
  Theory}, vol.~64, no.~12, pp. 7511--7528, 2018.

\bibitem{haslach2000efficient}
C.~Haslach and A.~Vinck, ``Efficient decoding of interleaved linear block
  codes,'' in \emph{IEEE International Symposium on Information Theory
  (ISIT)}.\hskip 1em plus 0.5em minus 0.4em\relax IEEE, 2000, p. 149.

\bibitem{Delsarte_1978}
P.~Delsarte, ``{Bilinear Forms over a Finite Field with Applications to Coding
  Theory},'' \emph{Journal of Combinatorial Theory, Series A}, vol.~25, no.~3,
  pp. 226--241, 1978.

\bibitem{Gabidulin_TheoryOfCodes_1985}
E.~M. Gabidulin, ``{Theory of Codes with Maximum Rank Distance},''
  \emph{Problems of Information Transmission}, vol.~21, no.~1, pp. 3--16, 1985.

\bibitem{Roth_RankCodes_1991}
R.~M. Roth, ``{Maximum-Rank Array Codes and their Application to Crisscross
  Error Correction},'' \emph{IEEE Transactions on Information Theory}, vol.~37,
  no.~2, pp. 328--336, Mar. 1991.

\bibitem{loidreau2006decoding}
P.~Loidreau and R.~Overbeck, ``Decoding rank errors beyond the error-correction
  capability,'' \emph{International Workshop on Algebraic and Combinatorial
  Coding Theory (ACCT)}, pp. 168--190, 2006.

\bibitem{silva2008rank}
D.~Silva, F.~R. Kschischang, and R.~Koetter, ``{A Rank-Metric Approach to Error
  Control in Random Network Coding},'' \emph{IEEE Transactions on Information
  Theory}, vol.~54, no.~9, pp. 3951--3967, 2008.

\bibitem{faure2006new}
C.~Faure and P.~Loidreau, ``{A New Public-Key Cryptosystem Based on the Problem
  of Reconstructing p--Polynomials},'' in \emph{Coding and Cryptography}.\hskip
  1em plus 0.5em minus 0.4em\relax Springer, 2006, pp. 304--315.

\bibitem{Gaborit-KeyRecoveryFaureLoidreau}
\BIBentryALTinterwordspacing
P.~Gaborit, A.~Otmani, and H.~{Tal\'e Kalachi}, ``{Polynomial-Time Key Recovery
  Attack on the {Faure--Loidreau} Scheme based on {G}abidulin Codes},''
  \emph{preprint}, Apr. 2017. [Online]. Available:
  \url{https://arxiv.org/abs/1606.07760}
\BIBentrySTDinterwordspacing

\bibitem{wachter2018repairing}
A.~{Wachter-Zeh}, S.~Puchinger, and J.~Renner, ``{Repairing the Faure--Loidreau
  Public-Key Cryptosystem},'' in \emph{IEEE International Symposium on
  Information Theory (ISIT)}, 2018.

\bibitem{renner2018rank}
J.~Renner, S.~Puchinger, and A.~Wachter-Zeh, ``{On a Rank-Metric Code-Based
  Cryptosystem with Small Key Size},'' \emph{submitted to: Transactions on
  Information Theory, arXiv preprint arXiv:1812.04892}, 2018.

\bibitem{sidorenko2010decoding}
V.~Sidorenko and M.~Bossert, ``{Decoding Interleaved Gabidulin Codes and
  Multisequence Linearized Shift-Register Synthesis},'' in \emph{IEEE
  International Symposium on Information Theory (ISIT)}, 2010, pp. 1148--1152.

\bibitem{gabidulin2000space}
E.~M. Gabidulin, M.~Bossert, and P.~Lusina, ``{Space-Time Codes Based on Rank
  Codes},'' in \emph{IEEE International Symposium on Information Theory}, 2000,
  p. 284.

\bibitem{lusina2002spacetime}
M.~{Bossert}, E.~M. {Gabidulin}, and P.~{Lusina}, ``{Space-Time Codes Based on
  Gaussian Integers},'' in \emph{IEEE International Symposium on Information
  Theory (ISIT)}, June 2002, pp. 273--.

\bibitem{liu2002rank}
Y.~Liu, M.~P. Fitz, and O.~Y. Takeshita, ``{A Rank Criterion for QAM Space-Time
  Codes},'' \emph{IEEE Transactions on Information Theory}, vol.~48, no.~12,
  pp. 3062--3079, 2002.

\bibitem{lusina2003}
P.~Lusina, E.~Gabidulin, and M.~Bossert, ``{Maximum Rank Distance Codes as
  Space--Time Codes},'' \emph{IEEE Transactions on Information Theory},
  vol.~49, no.~10, pp. 2757--2760, 2003.

\bibitem{robert2015new}
G.~Robert, ``{A New Constellation for Space-Time Coding},'' in
  \emph{International Workshop on Coding and Cryptography (WCC)}, 2015.

\bibitem{puchinger2016space}
S.~Puchinger, S.~Stern, M.~Bossert, and R.~F. Fischer, ``{Space-Time Codes
  Based on Rank-Metric Codes and Their Decoding},'' in \emph{IEEE International
  Symposium on Wireless Communication Systems}, 2016, pp. 125--130.

\bibitem{chabaud1996rsd}
F.~Chabaud and J.~Stern, ``{The Cryptographic Security of the Syndrome Decoding
  Problem for Rank Distance Codes},'' in \emph{International Conference on the
  Theory and Application of Cryptology and Information Security}, 1996, pp.
  368--381.

\bibitem{ourivski2002new}
A.~V. Ourivski and T.~Johansson, ``{New Technique for Decoding Codes in the
  Rank Metric and its Cryptography Applications},'' \emph{Problems of
  Information Transmission}, vol.~38, no.~3, pp. 237--246, 2002.

\bibitem{gaborit2016decoding}
P.~Gaborit, O.~Ruatta, and J.~Schrek, ``{On the Complexity of the Rank Syndrome
  Decoding Problem},'' \emph{IEEE Transactions on Information Theory}, vol.~62,
  no.~2, pp. 1006--1019, Feb 2016.

\bibitem{aragon:ISIT18}
N.~Aragon, P.~Gaborit, A.~Hauteville, and J.~Tillich, ``{A New Algorithm for
  Solving the Rank Syndrome Decoding Problem},'' in \emph{IEEE International
  Symposium on Information Theory (ISIT)}, June 2018, pp. 2421--2425.

\bibitem{wachter2014list}
A.~{Wachter-Zeh} and A.~Zeh, ``{List and Unique Error-Erasure Decoding of
  Interleaved {G}abidulin Codes with Interpolation Techniques},''
  \emph{Designs, Codes and Cryptography}, vol.~73, no.~2, pp. 547--570, 2014.

\bibitem{sidorenko2011skew}
V.~Sidorenko, L.~Jiang, and M.~Bossert, ``{Skew-Feedback Shift-Register
  Synthesis and Decoding Interleaved Gabidulin Codes},'' \emph{IEEE
  Transactions on Information Theory}, vol.~57, no.~2, pp. 621--632, 2011.

\bibitem{sidorenko2014fast}
V.~Sidorenko and M.~Bossert, ``{Fast Skew-Feedback Shift-Register Synthesis},''
  \emph{Designs, Codes and Cryptography}, vol.~70, no. 1-2, pp. 55--67, 2014.

\bibitem{puchinger2017row}
S.~Puchinger, J.~{Rosenkilde n{\'e} Nielsen}, W.~Li, and V.~Sidorenko, ``{Row
  Reduction Applied to Decoding of Rank-Metric and Subspace Codes},''
  \emph{Designs, Codes and Cryptography}, vol.~82, no. 1-2, pp. 389--409, 2017.

\bibitem{puchinger2017alekhnovich}
S.~Puchinger, S.~M{\"u}elich, D.~M{\"o}dinger, J.~{Rosenkilde n{\'e} Nielsen},
  and M.~Bossert, ``{Decoding Interleaved Gabidulin Codes Using Alekhnovich's
  Algorithm},'' \emph{Electronic Notes in Discrete Mathematics}, vol.~57, pp.
  175--180, 2017.

\bibitem{roth1996tensor}
R.~M. Roth, ``{Tensor Codes for the Rank Metric},'' \emph{IEEE Transactions on
  Information Theory}, vol.~42, no.~6, pp. 2146--2157, 1996.

\bibitem{kshevetskiy2005new}
A.~Kshevetskiy and E.~Gabidulin, ``{The New Construction of Rank Codes},'' in
  \emph{IEEE International Symposium on Information Theory (ISIT)}, 2005, pp.
  2105--2108.

\bibitem{augot2013rank}
D.~Augot, P.~Loidreau, and G.~Robert, ``{Rank Metric and Gabidulin Codes in
  Characteristic Zero},'' in \emph{IEEE International Symposium on Information
  Theory (ISIT)}, 2013.

\bibitem{gaborit2013low}
P.~Gaborit, G.~Murat, O.~Ruatta, and G.~Z{\'e}mor, ``{Low Rank Parity Check
  Codes and Their Application to Cryptography},'' in \emph{International
  Workshop on Coding and Cryptography (WCC)}, vol. 2013, 2013.

\bibitem{loidreau2016evolution}
P.~Loidreau, ``{An Evolution of GPT Cryptosystem},'' in \emph{International
  Workshop on Algebraic and Combinatorial Coding Theory (ACCT)}, 2016.

\bibitem{sheekey2016new}
J.~Sheekey, ``{A New Family of Linear Maximum Rank Distance Codes},''
  \emph{Advances in Mathematics of Communications}, vol.~10, no.~3, pp.
  475--488, 2016.

\bibitem{otal2016explicit}
K.~Otal and F.~{\"O}zbudak, ``{Explicit Construction of Some Non-Gabidulin
  Linear Maximum Rank Distance Codes},'' \emph{Advances in Mathematics of
  Communications}, vol.~10, no.~3, 2016.

\bibitem{lunardon2015generalized}
G.~Lunardon, R.~Trombetti, and Y.~Zhou, ``{Generalized Twisted Gabidulin
  codes},'' \emph{arXiv preprint arXiv:1507.07855}, 2015.

\bibitem{puchinger2017further}
S.~Puchinger, {Rosenkilde n\'e Nielsen}, and J.~Sheekey, ``{Further
  Generalisations of Twisted Gabidulin Codes},'' in \emph{International
  Workshop on Coding and Cryptography (WCC)}, 2017.

\bibitem{horlemann2015new}
A.-L. {Horlemann-Trautmann} and K.~Marshall, ``{New Criteria for MRD and
  Gabidulin Codes and Some Rank-Metric Code Constructions},'' \emph{Advances in
  Mathematics of Communications}, vol.~11, no.~3, pp. 533--548, 2017.

\bibitem{wachter2013decoding}
A.~{Wachter-Zeh}, ``{Decoding of Block and Convolutional Codes in Rank
  Metric},'' Ph.D. dissertation, Ulm University and Universit{\'e} Rennes 1,
  2013.

\bibitem{matsaglia1974}
G.~Matsaglia and G.~P.~H. Styan, ``{Equalities and Inequalities for Ranks of
  Matrices},'' \emph{Linear and Multilinear Algebra}, vol.~2, no.~3, pp.
  269--292, 1974.

\bibitem{gabidulin2008error}
E.~M. Gabidulin and N.~I. Pilipchuk, ``{Error and Erasure Correcting Algorithms
  for Rank Codes},'' \emph{Designs, codes and Cryptography}, vol.~49, no. 1-3,
  pp. 105--122, 2008.

\bibitem{silva2009fast}
D.~Silva and F.~R. Kschischang, ``{Fast Encoding and Decoding of Gabidulin
  Codes},'' in \emph{IEEE International Symposium on Information Theory
  (ISIT)}, 2009, pp. 2858--2862.

\bibitem{wachter2013fast}
A.~Wachter-Zeh, V.~Afanassiev, and V.~Sidorenko, ``{Fast Decoding of Gabidulin
  Codes},'' \emph{Designs, Codes and Cryptography}, vol.~66, no. 1-3, pp.
  57--73, 2013.

\bibitem{puchinger2018fast}
S.~Puchinger and A.~Wachter-Zeh, ``{Fast Operations on Linearized Polynomials
  and Their Applications in Coding Theory},'' \emph{Journal of Symbolic
  Computation}, vol.~89, pp. 194--215, 2018.

\bibitem{wardlaw_1994}
W.~P. Wardlaw, ``{Matrix Representation of Finite Fields},'' \emph{Mathematics
  Magazine}, vol.~67, no.~4, pp. 289--293, 1994.

\bibitem{couveignes2009elliptic}
J.-M. Couveignes and R.~Lercier, ``{Elliptic Periods for Finite Fields},''
  \emph{{Finite Fields and Their Applications}}, vol.~15, no.~1, pp. 1--22,
  2009.

\bibitem{sagemath}
{The Sage Developers}, \emph{{S}ageMath, the {S}age {M}athematics {S}oftware
  {S}ystem}, 2019, {\tt https://www.sagemath.org}.

\bibitem{Overbeck_Diss_InterleveadGab}
R.~Overbeck, ``{Public Key Cryptography based on Coding Theory},'' Ph.D.
  dissertation, TU Darmstadt, Darmstadt, Germany, 2007.

\bibitem{sidorenko2008decoding}
V.~Sidorenko, G.~Schmidt, and M.~Bossert, ``{Decoding Punctured Reed--Solomon
  Codes up to the Singleton Bound},'' in \emph{International ITG Conference on
  Source and Channel Coding}.\hskip 1em plus 0.5em minus 0.4em\relax VDE, 2008.

\bibitem{holzbaur2019lrc}
L.~Holzbaur, S.~Puchinger, and A.~Wachter-Zeh, ``{On Error Decoding of Locally
  Repairable and Partial MDS Codes},'' \emph{arXiv preprint arXiv:1904.05623},
  2019.

\bibitem{caruso2017fast}
X.~Caruso and J.~Le~Borgne, ``{Fast Multiplication for Skew Polynomials},'' in
  \emph{Proceedings of the 2017 ACM on International Symposium on Symbolic and
  Algebraic Computation}.\hskip 1em plus 0.5em minus 0.4em\relax ACM, 2017, pp.
  77--84.

\bibitem{mceliece1978public}
R.~J. McEliece, ``{A Public-Key Cryptosystem Based on Algebraic Coding
  Theory},'' \emph{Coding Thv}, vol. 4244, pp. 114--116, 1978.

\bibitem{gabidulin1991ideals}
E.~M. Gabidulin, A.~Paramonov, and O.~Tretjakov, ``{Ideals over a
  Non-Commutative Ring and Their Application in Cryptology},'' in
  \emph{Workshop Theory and Appl. Cryptogr. Techn.}\hskip 1em plus 0.5em minus
  0.4em\relax Springer, 1991, pp. 482--489.

\end{thebibliography}

\end{document}